\documentclass[smallextended]{svjour3}       
\smartqed  
\usepackage{graphicx}
\usepackage{makeidx}
\usepackage{color}
\usepackage{amsfonts}
\usepackage{cite}
\usepackage{url}
\usepackage{epsfig,dsfont,amsmath,amssymb}
\usepackage{wrapfig,floatflt}
\usepackage{theorem}
\usepackage{graphics}
\usepackage{multirow}
\usepackage{footnote}
\usepackage{hyperref}
\makesavenoteenv{tabular}
\graphicspath{{./images/}}

\def \bs {\boldsymbol}
\def \Pr {\mathbb{P}}

\def \I {\mathrm{I}}
\def \R {\mathbb{R}}

\def \e {\mathbf{e}}
\def \a {\bs{a}}
\def \sF {\mathcal{F}}
\def \X {\bs{X}}

\def \y {\bs y}
\def \x {\bs x}

\def \hx {\hat{\x}}
\def \z {\bs z}
\def \w {\bs w}
\def \h {\bs h}
\def \H {\mathcal{H}}

\def \u {\bs u}
\def \v {\bs v}
\def \g {\bs g}

\def \zero {\bs 0}
\def \N {\mathcal N}
\def \E {\mathbb E}

\def \df {\stackrel{\mathrm{def}}{=}}

\def \bsupp {\mathrm{bsupp}}

\def \bone {\bs 1}

\def \sup {\mathrm{sup}}

\def \nn {\nonumber}
\def \ker {\mathrm{Ker}}
\def \b {\mathrm{b}}

\def \sI {\mathcal{I}}

\definecolor{Dblue}{rgb}{0,0,1}
\definecolor{Dbrown}{rgb}{0.59,0.4,0}
\definecolor{Dred}{rgb}{0.64,0,0}
\definecolor{Dgreen}{rgb}{0,0.4,0}

\begin{document}

\title{Fixed point theory and semidefinite programming for computable performance analysis of block-sparsity recovery\thanks{This work was supported by ONR Grant N000140810849, and NSF Grants CCF-1014908 and CCF-0963742.}\\
}

\titlerunning{Computable performance analysis of block-sparsity recovery}        

\author{Gongguo Tang         \and
        Arye Nehorai 
}


\institute{Gongguo Tang \at
               Washington University in St. Louis, St. Louis, MO 63130-1127 \\
              \email{gt2@ese.wustl.edu}           
           \and
           Arye Nehorai \at
                          Washington University in St. Louis, St. Louis, MO 63130-1127 \\
                          \email{nehorai@ese.wustl.edu}
}

\date{Received: date / Accepted: date}

\maketitle

\begin{abstract}
In this paper, we employ fixed point theory and semidefinite programming to compute the performance bounds on convex block-sparsity recovery algorithms. As a prerequisite for optimal sensing matrix design, a computable performance bound would open doors for wide applications in sensor arrays, radar, DNA microarrays, and many other areas where block-sparsity arises naturally. We define a family of goodness measures for arbitrary sensing matrices as the optimal values of certain optimization problems. The reconstruction errors of convex recovery algorithms are bounded in terms of these goodness measures. We demonstrate that as long as the number of measurements is relatively large, these goodness measures are bounded away from zero for a large class of random sensing matrices, a result parallel to the probabilistic analysis of the block restricted isometry property. As the primary contribution of this work, we associate the goodness measures with the fixed points of functions defined by a series of semidefinite programs. This relation with fixed point theory yields efficient algorithms with global convergence guarantees to compute the goodness measures.
\keywords{Block-sparsity recovery \and Compressive sensing \and Fixed point theory \and Semidefinite programming}
\subclass{47H10\and 09C22\and 90C90\and 94A12}
\end{abstract}

%

\section{Introduction}
\label{sec:intro}
\noindent

Recovery of signals with low-dimensional structures, in particular, sparsity \cite{Candes2008IntroCS}, block-sparsity \cite{Eldar2009BlockSparse}, and low-rankness \cite{candes2009lowrank}, has found numerous applications in signal sampling, control, inverse imaging, remote sensing, radar, sensor arrays, image processing, computer vision, and so on. Mathematically, the recovery of signals with low-dimensional structures aims to reconstruct a signal with a prescribed structure, usually from noisy linear measurements, as follows:
\begin{eqnarray}
  \y &=& A\x + \w,
\end{eqnarray}
where $\x \in \R^N$ is the signal to be reconstructed, $\y \in \R^m$ is the measurement vector, $A \in \R^{m\times N}$ is the sensing/measurement matrix, and $\w \in \R^m$ is the noise. For example, a sparse signal is assumed to have only a few non-zero coefficients when represented as a linear combination of atoms from an orthogonal basis or from an overcomplete dictionary. Similarly, for a block-sparse signal, the non-zero coefficients are assumed to cluster into blocks.

A theoretically justified way to exploit the low-dimensional structure in recovering $\x$ is to minimize a convex function that is known to enforce that low-dimensional structure. Examples include using the $\ell_1$ norm to enforce sparsity, block-$\ell_1$ norm (or $\ell_2/\ell_1$ norm) to enforce the block-sparsity, and the nuclear norm to enforce the low-rankness. The performance of these convex enforcements is usually analyzed using variants of the restricted isometry property (RIP) \cite{Candes2008RIP, Eldar2009BlockSparse, candes2009lowrank}. Upper bounds on the $\ell_2$ norm of the error vectors for various recovery algorithms have been expressed in terms of the RIP. Unfortunately, it is extremely difficult to verify that the RIP of a specific sensing matrix satisfies the conditions for the bounds to be valid, and even more difficult to directly compute the RIP itself. Actually, the only known sensing matrices with nice RIPs are certain types of random matrices \cite{Juditsky2010Verifiable}.

In this paper, we investigate the recovery performance for block-sparse signals. Block-sparsity arises naturally in applications such as sensor arrays \cite{Willsky2005Source}, radar \cite{Sen2011Multi}, multi-band signals \cite{mishali2009blind}, and DNA microarrays \cite{Parvaresh2008DNA}. A particular area that motivates this work is the application of block-sparse signal recovery in radar systems. The signals in radar applications are usually sparse because there are only a few targets to be estimated among many possibilities. However, a single target manifests itself simultaneously in the sensor domain, the frequency domain, the temporal domain, and the reflection-path domain. As a consequence, the underlying signal would be block-sparse when the radar system observes the targets from several of these domains \cite{sen2011ofdm, Sen2011Multi}.

The aforementioned applications require a computable performance analysis of block-sparsity recovery. While it is perfectly reasonable to use a random sensing matrix for signal sampling, the sensing matrices in other applications are far from random and actually depend on the underlying measurement devices and the physical processes that generate the observations. Due to the computational challenges associated with the RIP, it is necessary to seek computationally more amenable goodness measures of the sensing matrices. Computable performance measures would open doors for wide applications. They would provide a means to pre-determine the performance of the sensing system before its implementation and the taking of measurements. In addition, in applications where we have the flexibility to choose the sensing matrix, computable performance analysis would form a bisis for optimal sensing matrix design.

We preview our contributions. First of all, we define a family of goodness measures of the sensing matrix, and use them to derive performance bounds on the block-$\ell_\infty$ norm of the recovery error vector. Performance bounds using other norms are expressed using the block-$\ell_\infty$ norm. Our preliminary numerical results show that these bounds are tighter than the block RIP based bounds. Second and most important, we develop a fixed point iteration framework to design algorithms that efficiently compute lower bounds on the goodness measures for arbitrary sensing matrices. Each fixed point iteration solves a series of semidefinite programs. The fixed point iteration framework also demonstrates the algorithms' convergence to the global optima from any initial point. As a by-product, we obtain a fast algorithm to verify the sufficient condition guaranteeing exact block-sparsity recovery via block-$\ell_1$ minimization. Finally, we show that the goodness measures are non-degenerate for subgaussian and isotropic random sensing matrices as long as the number of measurements is relatively large, a result parallel to that of the block RIP for random matrices.

This work extends verifiable and computable performance analysis from sparse signal cases \cite{tang2011linf, Juditsky2010Verifiable, dAspremont2007sparsePCA, dAspermont2010Nullspace} to block-sparse signal cases. There are several technical challenges that are unique to the block-sparse case. In the sparse setting, the optimization subproblems are solved exactly by linear programming or second-order cone programming. However, in the block-sparse setting, the associated subproblems are not readily solvable and we need to develop semidefinite relaxations to compute an upper bound on the optimal values of the subproblems. The bounding technique complicates the situation as it is not clear how to obtain bounds on the goodness measures from these subproblem bounds. We develop a systematic fixed point theory to address this problem. 

The rest of the paper is organized as follows. In Section \ref{sec:model}, we introduce notations and we present the measurement model, three convex relaxation algorithms, and the sufficient and necessary condition for exact block-$\ell_1$ recovery. In section \ref{sec:bounds}, we derive performance bounds on the block-$\ell_\infty$ norms of the recovery errors for several convex relaxation algorithms. Section \ref{sec:random} is devoted to the probabilistic analysis of our block-$\ell_\infty$ performance measures. In Section \ref{sec:computation}, we design algorithms to verify a sufficient condition for exact block-$\ell_1$ recovery in the noise-free case, and to compute the goodness measures of arbitrary sensing matrices through fixed point iteration, bisection search, and semidefinite programming.  We evaluate the algorithms' performance in Section \ref{sec:numerical}. Section \ref{sec:conclusions} summarizes our conclusions.

\section{Notations, Measurement Model, and Recovery Algorithms}\label{sec:model}
In this section, we introduce notations and the measurement model, and present three block-sparsity recovery algorithms.

For any vector $\x \in \R^{np}$, we partition the vector into $p$ blocks, each of length $n$. More precisely, we have $\x = \left[\x_1^T, \ldots, \x_p^T\right]^T$, with the $i$th block $\x_i \in \R^n$. The block-$\ell_q$ norms for $1 \leq q \leq \infty$ associated with this block-sparse structure are defined as
\begin{eqnarray}
  \|\x\|_{\b q} &=& \left(\sum_{i=1}^p \|\x_i\|_2^q\right)^{1/q},\ 1\leq q < \infty
\end{eqnarray}
and
\begin{eqnarray}
  \|\x\|_{\b \infty} &=& \max_{1\leq i\leq p} \|\x_i\|_2,\ q = \infty.
\end{eqnarray}
The canonical inner product in $\R^{np}$ is denoted by $\left<\cdot, \cdot\right>$, and the $\ell_2$ (or Euclidean) norm is $\|\x\|_2 = \sqrt{\left<\x,\x\right>}$. Obviously, the block-$\ell_2$ norm is the same as the ordinary $\ell_2$ norm. We use $\|\cdot\|_\diamond$ to denote a general norm, and use $\otimes$ to denote the Kronecker product.

The block support of $\x \in \R^{np}$, $\bsupp(\x) = \{i: \|\x_i\|_2 \neq 0\}$, is the index set of the non-zero blocks of $\x$. The size of the block support, denoted by the block-$\ell_0$ ``norm" $\|\x\|_{\b 0}$,  is the block-sparsity level of $\x$. Signals of block-sparsity level of at most $k$ are called $k-$block-sparse signals. If $S \subset \{1,\cdots,p\}$ is an index set, then $|S|$ is the cardinality of $S$, and $\x_S \in \R^{n|S|}$ is the vector formed by the blocks of $\x$ with indices in $S$.

We use $\e_i$, $\zero$, $\bs O$, $\bone$, and $\I_N$ to denote respectively the $i$th canonical basis vector, the zero column vector, the zero matrix, the column vector with all ones, and the identity matrix of size $N \times N$.

Suppose $\x$ is a $k-$block-sparse signal. In this paper, we observe $\x$ through the following linear model:
\begin{eqnarray}\label{eqn:model}
  \y &=& A\x + \w,
\end{eqnarray}
where $A \in \R^{m\times np}$ is the measurement/sensing matrix, $\y$ is the measurement vector, and $\w$ is noise. A very special block-sparse model is when $\x$ is a complex signal, such as the models in sensor arrays and radar applications. Note that the computable performance analysis developed in \cite{tang2011cmsv, tang2011linf} for the real variables can not apply to the complex case.

Many algorithms in sparse signal recovery have been extended to recover the block-sparse signal $\x$ from $\y$ by exploiting the block-sparsity of $\x$. We focus on three algorithms based on block-$\ell_1$ minimization: the Block-Sparse Basis Pursuit (BS-BP) \cite{Eldar2009BlockSparse}, the Block-Sparse Dantzig selector (BS-DS) \cite{zhang2010groupdantzig}, and the Block-Sparse LASSO estimator (BS-LASSO) \cite{Yuan2006Grouplasso}.
\begin{eqnarray}
\hskip -1cm &&\text{BS-BP:}\min_{\z \in \R^{np}}\|\z\|_{\b 1} \text{\ \ s.t.\ } \|\y - A\z\|_2 \leq \varepsilon\label{bp}\\
\hskip -1cm && \text{BS-DS:} \min_{\z \in \R^{np}}\|\z\|_{\b 1} \text{\ \ s.t. \ } \|A^T(\y - A\z)\|_{\b \infty} \leq \mu\label{ds}\\
\hskip -1cm &&\text{BS-LASSO:} \min_{\z \in \R^{np}} \frac{1}{2}\|\y - A\z\|_2^2 + \mu \|\z\|_{\b 1} \label{lasso}.
\end{eqnarray}
Here $\mu$ is a tuning parameter, and $\varepsilon$ is a measure of the noise level. All three optimization problems have efficient implementations using convex programming.

In the noise-free case where $\w = 0$, roughly speaking, all three algorithms reduce to
\begin{eqnarray}\label{eqn:l1relaxation}
\min_{\z \in \R^{np}}\|\z\|_{\b 1} \text{\ \ s.t.\ } A\z = A\x,
\end{eqnarray}
which is the block-$\ell_1$ relaxation of the block-$\ell_0$ problem:
\begin{eqnarray}
  \min_{\z \in \R^n}\|\z\|_{\b 0} \text{\ \ s.t.\ } A\z = A\x.
\end{eqnarray}
A minimal requirement on the block-$\ell_1$ minimization algorithms is the \emph{uniqueness and exactness} of the solution $\hx \df \mathrm{argmin}_{\z: A\z = A\x} \|\x\|_{\b 1}$, \emph{i.e.}, $\hx = \x$. When the true signal $\x$ is $k-$block-sparse, the sufficient and necessary condition for exact block-$\ell_1$ recovery is \cite{Stojnic2009BS}
\begin{eqnarray}\label{nullspaceproperty}
  \sum_{i\in S}\|\z_i\|_2 < \sum_{i \notin S} \|\z_i\|_2, \forall \z \in \ker(A), |S| \leq k,
\end{eqnarray}
where $\ker(A) \df \{\z: A\z = 0\}$ is the kernel of $A$, $S \subset \{1,\ldots,p\}$ is an index set, and $\z_i$ is the $i$th block of $\z$ of size $n$.

\section{Performance Bounds on the Block-$\ell_\infty$ Norms of the Recovery Errors}\label{sec:bounds}
In this section, we derive performance bounds on the block-$\ell_\infty$ norms of the error vectors. We first establish a proposition characterizing the error vectors of the block-$\ell_1$ recovery algorithms, whose proof is given in Appendix \ref{app:pf:errorcharacteristics}.

\begin{proposition}\label{pro:errorcharacteristics}
Suppose $\x \in \R^{np}$ in \eqref{eqn:model} is $k-$block-sparse and the noise $\w$ satisfies $\|\w\|_2 \leq \varepsilon$, $\|A^T\w\|_{\b\infty} \leq \mu$, and $\|A^T\w\|_\infty \leq \kappa \mu, \kappa \in (0,1)$, for the BS-BP, the BS-DS, and the BS-LASSO, respectively. Define $\h = \hx - \x$ as the error vector for any of the three block-$\ell_1$ recovery algorithms \eqref{bp}, \eqref{ds}, and \eqref{lasso}. Then we have
\begin{eqnarray}\label{eqn:errorcharacteristics}
\|\h_S\|_{\b 1} \geq  \|\h\|_{\b 1}/c,
\end{eqnarray}
where $S = \mathrm{bsupp}(\x)$, $ c = 2$ for the BS-BP and the BS-DS, and $c = 2/(1-\kappa)$ for the BS-LASSO.
\end{proposition}

An immediate corollary of Proposition \ref{pro:errorcharacteristics} is to bound the block-$\ell_1$ and $\ell_2$ norms of the error vectors, using the block-$\ell_\infty$ norm. The proof is given in Appendix \ref{app:pf:connections}.
\begin{corollary}\label{cor:connections}
Under the assumptiocorollaryns of Proposition \ref{pro:errorcharacteristics}, we have
\begin{eqnarray}
  \|\h\|_{\b 1} & \leq & ck \|\h\|_{\b \infty}, \text{\ and \ }\label{eqn:l1linf}\\
  \|\h\|_2 & \leq & \sqrt{ck} \|\h\|_{\b \infty}.\label{eqn:l2linf}
\end{eqnarray}
Furthermore, if $S = \bsupp(\x)$ and $\beta = \min_{i\in S}\|\x_i\|_2$, then $\|\h\|_{\b \infty} < \beta/2$ implies
\begin{eqnarray}
  \{i: \|\hx_i\|_2 > \beta/2\} &=& \bsupp(\x),
\end{eqnarray}
\emph{i.e.}, a thresholding operator recovers the signal block-support.
\end{corollary}

For ease of presentation, we introduce the following notation:
\begin{definition}\label{def:linfcmsv}
For any $s \in [1,p]$ and matrix $A\in \R^{m\times np}$, define
\begin{eqnarray}
  \omega_{\diamond}(Q,s) &=& \min_{\z: \|\z\|_{\b 1}/\|\z\|_{\b \infty} \leq s} \frac{\|Q\z\|_\diamond}{\|\z\|_{\b \infty}},
\end{eqnarray}
where $Q$ is either $A$ or $A^TA$.
\end{definition}

Now we present the error bounds on the block-$\ell_{\infty}$ norm of the error vectors for the BS-BP, the BS-DS, and the BS-LASSO, whose proof is given in Appendix \ref{app:pf:errorbound}.
\begin{theorem}\label{thm:errorbound}
Under the assumption of Proposition \ref{pro:errorcharacteristics}, we have
\begin{eqnarray}
  \|\hx - \x\|_{\b \infty}\leq \frac{2\varepsilon}{\omega_2(A,2k)}
\end{eqnarray}
for the BS-BP,
\begin{eqnarray}
  \|\hx-\x\|_{\b \infty}\leq \frac{2\mu}{\omega_{\b\infty}(A^TA,2k)}
\end{eqnarray}
for the BS-DS, and
\begin{eqnarray}
  \|\hx-\x\|_{\b \infty}\leq \frac{(1+\kappa)\mu}{\omega_{\b\infty}(A^TA,2k/(1-\kappa))}
\end{eqnarray}
for the BS-LASSO.
\end{theorem}

A consequence of Theorem \ref{thm:errorbound} and Corollary \ref{cor:connections} is the error bound on the $\ell_2$ norm:
\begin{corollary}\label{cor:l2errorbound}
Under the assumption of Proposition \ref{pro:errorcharacteristics}, the $\ell_2$ norms of the recovery errors are bounded as
\begin{eqnarray}
  \|\hx - \x\|_2\leq \frac{2\sqrt{2k}\varepsilon}{\omega_2(A,2k)}
\end{eqnarray}
for the BS-BP,
\begin{eqnarray}
  \|\hx-\x\|_2\leq \frac{2\sqrt{2k}\mu}{\omega_{\b\infty}(A^TA,2k)}
\end{eqnarray}
for the BS-DS, and
\begin{eqnarray}
  \|\hx-\x\|_2\leq \sqrt{\frac{2k}{1-\kappa}}\frac{(1+\kappa)\mu}{\omega_{\b\infty}(A^TA,2k/(1-\kappa))}
\end{eqnarray}
for the BS-LASSO.
\end{corollary}

One of the primary contributions of this work is the design of algorithms that compute $\omega_\diamond(A,s)$ and $\omega_{\b \infty}(A^TA,s)$ efficiently. The algorithms provide a way to numerically assess the performance of the BS-BP, the BS-DS, and the BS-LASSO according to the bounds given in Theorem \ref{thm:errorbound} and Corollary \ref{cor:l2errorbound}. According to Corollary \ref{cor:connections}, the correct recovery of signal block-support is also guaranteed by reducing the block-$\ell_\infty$ norm to some threshold. In Section \ref{sec:random}, we also demonstrate that the bounds in Theorem \ref{thm:errorbound} are non-trivial for a large class of random sensing matrices, as long as $m$ is relatively large. Numerical simulations in Section \ref{sec:numerical} show that in many cases the error bounds on the $\ell_2$ norms based on Corollary \ref{cor:l2errorbound} are tighter than the block RIP based bounds. Before we turn to the computation issues, we first establish some results on the probability behavior of $\omega_\diamond(Q,s)$.

\section{Probabilistic Behavior of $\omega_\diamond(Q,s)$}\label{sec:random}
In this section, we analyze how good are the performance bounds in Theorem \ref{thm:errorbound} for random sensing matrices. For this purpose, we define the block $\ell_1$-constrained minimal singular value (block $\ell_1$-CMSV), which is an extension of the $\ell_1$-CMSV concept in the sparse setting \cite{tang2011cmsv}:
\begin{definition}\label{def:l1cmsv}
For any $s \in [1,p]$ and matrix $A\in \R^{m\times np}$, define the block $\ell_1$-constrained minimal singular value (abbreviated as block $\ell_1$-CMSV) of $A$ by
\begin{eqnarray}
\rho_s(A) = \min_{\z:\ {\|\z\|_{\b 1}^2}/{\|\z\|_2^2} \leq s} \frac{\|A\z\|_2}{\|\z\|_2}.
\end{eqnarray}
\end{definition}

The most important difference between $\rho_s(A)$ and $\omega_\diamond(Q,s)$ is the replacement of $\|\cdot\|_{\b \infty}$ with $\|\cdot\|_2$ in the denominators of the fractional constraint and the objective function. The Euclidean norm $\|\cdot\|_2$ is more amenable to probabilistic analysis. The connections between $\rho_s(A)$, $\omega_2(A,s)$, and $\omega_{\b\infty}(A^TA,s)$, established in the following lemma, allow us to analyze the probabilistic behaviors of $\omega_\diamond(Q,s)$ using the results for $\rho_s(A)$ which we are going to establish later.
\begin{lemma}\label{lm:omega_rho}
\begin{eqnarray}
   \sqrt{s}\sqrt{\omega_{\b \infty}(A^TA,s)} \geq \omega_2(A,s) \geq \rho_{s^2}(A).
\end{eqnarray}
\end{lemma}

For a proof, see Appendix \ref{app:pf:lm_omega_rho}.

Next we derive a condition on the number of measurements needed to get $\rho_s(A)$ bounded away from zero with high probability for sensing matrices with \emph{i.i.d.} subgaussian and isotropic rows. Note that a random vector $\X \in \R^{np}$ is called \emph{isotropic and subgaussian} with constant $L$ if $\E|\left<\X, \u\right>|^2 = \|\u\|_2^2$ and $\Pr(|\left<\X, \u\right>| \geq t ) \leq 2 \exp(-t^2/(L \|\u\|_2))$ hold for any $\u \in \R^{np}$.

\begin{theorem}\label{thm:randomcmsv}
Let the rows of $\sqrt{m}A$ be \emph{i.i.d.} subgaussian and isotropic random vectors with numerical constant $L$. Then there exist constants $c_1$ and $c_2$ such that for any $\epsilon > 0$ and $m \geq 1$ satisfying
\begin{eqnarray}
  m \geq c_1 \frac{L^2(sn + s \log p)}{\epsilon^2},
\end{eqnarray}
we have
\begin{eqnarray}
  \E \rho_s(A) \geq 1- \epsilon,
\end{eqnarray}
and
\begin{eqnarray}
  \Pr\{\rho_s(A) \geq 1 - \epsilon\} \geq  1 - \exp(-c_2 \epsilon^2m/L^4).
\end{eqnarray}
\end{theorem}

For a proof, see Appendix \ref{app:pf:randomcmsv}.

Using $\rho_s(A)$, we could equally develop bounds similar to those of Theorem \ref{thm:errorbound} on the $\ell_2$ norm of the error vectors. For example, the error bound for the BS-BP would look like
\begin{eqnarray}
  \|\hx - \x\|_2 & \leq & \frac{2\varepsilon}{\rho_{2k}(A)}.
\end{eqnarray}
The conclusion of Theorem \ref{thm:randomcmsv}, combined with the previous equation, implies that we could stably recover a block sparse signal using BS-BP with high probability if the sensing matrix is subgaussian and isotropic and $m \geq c (kn + k\log p)/\epsilon^2$. If we do not consider the block structure in the signal, we would need $m  \geq c (kn\log (np))/\epsilon^2$ measurements, as the sparsity level is $kn$ \cite{tang2011cmsv}. Therefore, the prior information regarding the block structure greatly reduces the number of measurements necessary to recover the signal. The lower bound on $m$ is essentially the same as the one given by the block RIP (See \cite[Proposition 4]{Eldar2009BlockSparse}).

\begin{theorem}\label{thm:randomomega}
Under the assumptions and notations of Theorem \ref{thm:randomcmsv}, there exist constants $c_1$ and $c_2$ such that for any $\epsilon > 0$ and $m \geq 1$ satisfying
\begin{eqnarray}\label{eqn:mrandombd}
  m \geq c_1 \frac{L^2(s^2n + s^2 \log p)}{\epsilon^2},
\end{eqnarray}
we have
\begin{eqnarray}
  &&\E\ \omega_2(A,s) \geq 1 - \epsilon, \text{\ and \ }\\
  &&\Pr\{\omega_2(A,s) \geq 1 - \epsilon\} \geq  1 - \exp(-c_2 \epsilon^2m/L^4),
\end{eqnarray}
and
\begin{eqnarray}
  \hskip -1cm &&\E {\ \omega_{\b \infty}(A^TA,s)} \geq \frac{(1-\epsilon)^2}{s}, \text{\ and \ }\label{eqn:meanomega_inf}\\
  \hskip -1cm &&\Pr\left\{\omega_{\b \infty}(A^TA,s) \geq \frac{(1 - \epsilon)^2}{s}\right\} \geq  1 - \exp(-c_2 \epsilon^2m/L^4)\label{eqn:probomega_inf}.
\end{eqnarray}
\end{theorem}

Equation \eqref{eqn:mrandombd} and Theorem \ref{thm:errorbound} imply that for exact signal recovery in the noise-free case, we need $O(s^2(n + \log p))$ measurements for random sensing matrices. The extra $s$ suggests that the $\omega_\diamond$ based approach to verify exact recovery is not as good as the one based on $\rho_s$. However, $\omega_\diamond$ is computationally more amenable, as we are going to see in Section \ref{sec:computation}. The measurement bound \eqref{eqn:mrandombd} also implies that the algorithms for verifying $\omega_\diamond > 0$ and for computing $\omega_\diamond$ work for $s$ at least up to the order $\sqrt{m/(n + \log p)}$.

Finally, we comment that sensing matrices with \emph{i.i.d.} subgaussian and isotropic rows include the Gaussian ensemble, the Bernoulli ensemble, and normalized volume measure on various convex symmetric bodies, for example, the unit balls of $\ell_q^{np}$ for $2 \leq q \leq \infty$ \cite{mendelson2007subgaussian}. 

\section{Verification and Computation of $\omega_\diamond$}\label{sec:computation}
In this section, we consider the computational issues of $\omega_\diamond(\cdot)$. We will present a very general algorithm and make it specific only when necessary. For this purpose, we use $Q$ to denote either $A$ or $A^TA$. 

\subsection{Verification of $\omega_\diamond > 0$}
A prerequisite for the bounds in Theorem \ref{thm:errorbound} to be valid is the positiveness of the involved $\omega_\diamond(\cdot)$. We call the validation of $\omega_\diamond(\cdot) > 0$ the verification problem. Note that from Theorem \ref{thm:errorbound}, $\omega_\diamond(\cdot) > 0$ implies the exact recovery of the true signal $\x$ in the noise-free case. Therefore, verifying $\omega_\diamond(\cdot) > 0$ is equivalent to verifying a sufficient condition for exact block-$\ell_1$ recovery.

Verifying $\omega_\diamond(Q,s) > 0$ amounts to making sure $\|\z\|_{\b 1}/\|\z\|_{\b \infty}\leq s$ for all $\z$ such that $Q\z = 0$. Therefore, we compute
\begin{eqnarray}\label{eqn:s_star_larger}
  s^* &=& \min_{\z} \frac{\|\z\|_{\b 1}}{\|\z\|_{\b\infty}} \text{\ s.t. \ } Q\z = 0.
\end{eqnarray}
Then, when $s < s^*$, we have $\omega_\diamond(Q,s) > 0$. The following theorem presents an optimization procedure that computes a lower bound on $s^*$.
\begin{proposition}\label{pro:max_sparse}
The reciprocal, denoted by $s_*$, of the optimal value of the following optimization, 
\begin{eqnarray}\label{eqn:max_sparse}
 \max_i\min_{P_i}\max_{j}\|\delta_{ij}\I_n - P_i^TQ_j\|_2
\end{eqnarray}
is a lower bound on $s^*$. Here $P$ is a matrix variable of the same size as $Q$, $\delta_{ij} = 1$ for $i = j$ and $0$ otherwise, and $P = \left[P_1,\ldots,P_p\right]$, $Q = \left[Q_1,\ldots,Q_p\right]$ with $P_i$ and $Q_j$ having $n$ columns each.
\end{proposition}

The proof is given in Appendix \ref{app:pf:max_sparse}.

Because $s_* < s^*$, the condition $s < s_*$ is a sufficient condition for $\omega_\diamond > 0$ and for the uniqueness and exactness of block-sparse recovery in the noise-free case. To get $s_*$, for each $i$, we need to solve
\begin{eqnarray}\label{eqn:max_sparse_i}
\min_{P_i}\max_{j}\|\delta_{ij}\I_n - P_i^TQ_j\|_2.
\end{eqnarray}
A semidefinite program equivalent to \eqref{eqn:max_sparse_i} is given as follows:
\begin{eqnarray}
&&\min_{P_i, t} t \text{\ s.t.\ } \|\delta_{ij}\I_n - P_i^T Q_j\|_2 \leq t, j = 1, \ldots, p, \label{eqn:max_sparse_cvx1}\\
&\Leftrightarrow & \min_{P_i, t} t \text{\ s.t.\ } \left[
                                                     \begin{array}{cc}
                                                       t\I_n & \delta_{ij}\I_n - P_i^TQ_j \\
                                                       \delta_{ij}I_n -Q_j^T P_i & t\I_n \\
                                                     \end{array}
                                                   \right]\succeq 0, j = 1, \ldots, p.\label{eqn:max_sparse_cvx2}
\end{eqnarray}
Small instances of \eqref{eqn:max_sparse_cvx1} and \eqref{eqn:max_sparse_cvx2} can be solved using CVX \cite{Grant2009CVX}. However, for medium to large scale problems, it is beneficial to use first-order techniques to solve \eqref{eqn:max_sparse_i} directly. We observe that $\max_{j}\|\delta_{ij}\I_n - P_i^TQ_j\|_2$ can be expressed as the largest eigenvalue of a block-diagonal matrix. The smoothing technique for semidefinite optimization developed in \cite{Nesterov2007Smoothing} can be used to minimize the largest eigenvalue with respect to $P_i$. We leave these implementations in future work. 

Due to the equivalence of $A^TA\z = 0$ and $A\z = 0$, we always solve \eqref{eqn:s_star_larger} for $Q = A$ and avoid $Q = A^TA$. The former apparently involves solving semidefinite programs of smaller size. In practice, we usually replace $A$ with the matrix with orthogonal rows obtained from the economy-size QR decomposition of $A^T$.

\subsection{Fixed Point Theory for Computing $\omega_\diamond(\cdot)$}\label{sec:fixedpoint}
We present a general fixed point procedure to compute $\omega_\diamond$. Recall that the optimization problem defining $\omega_\diamond$ is as follows:
\begin{eqnarray}\label{eqn:max_inf_Q_diamond}
  \omega_\diamond(Q,s) = \min_{\z} \frac{\|Q\z\|_\diamond}{\|\z\|_{\b \infty}} \text{\ s.t. \ } \frac{\|\z\|_{\b 1}}{\|\z\|_{\b \infty}} \leq s,
\end{eqnarray}
or equivalently,
\begin{eqnarray}\label{eqn:optimizerho}
  \frac{1}{\omega_\diamond(Q,s)} = \max_{\z} \|\z\|_{\b \infty}\text{\ s.t. \ } \|Q\z\|_\diamond \leq 1, \frac{\|\z\|_{\b 1}}{\|\z\|_{\b \infty}} \leq s.
\end{eqnarray}

For any $s \in (1, s^*)$, we define a function over $[0, \infty)$ parameterized by $s$
\begin{eqnarray}\label{def:fs}
  f_s(\eta) &=& \max_{\z} \left\{\|\z\|_{\b \infty}: \|Q\z\|_\diamond \leq 1, {\|\z\|_{\b 1}} \leq s \eta\right\}.
\end{eqnarray}
We basically replaced the $\|\z\|_{\b\infty}$ in the denominator of the fractional constraint in \eqref{eqn:optimizerho} with $\eta$. It turns out that the unique positive fixed point of $f_s(\eta)$ is exactly $1/\omega_\diamond(Q,s)$, as shown by the following proposition. See Appendix \ref{app:pf:fix_fs} for the proof.

\begin{proposition}\label{pro:fix_fs}
The function $f_s(\eta)$ has the following properties:
\begin{enumerate}
  \item $f_s(\eta)$ is continuous in $\eta$.
  \item $f_s(\eta)$ is strictly increasing in $\eta$.
\item $f_s(0) = 0$, $f_s(\eta) \geq s\eta > \eta$ for sufficiently small $\eta > 0$, and there exists $\rho < 1$ such that $f_s(\eta) < \rho\eta$ for sufficiently large $\eta$.
\item $f_s(\eta)$ has a unique positive fixed point $\eta^* = f_s(\eta^*)$ that is equal to $1/\omega_\diamond(Q,s)$.
\item For $\eta \in (0, \eta^*)$, we have $f_s(\eta) > \eta$,  and for $\eta \in (\eta^*, \infty)$, we have $f_s(\eta) < \eta$.
\item For any $\epsilon > 0$, there exists $\rho_1(\epsilon) > 1$ such that $f_s(\eta) > \rho_1(\epsilon) \eta$ as long as $0 < \eta < (1-\epsilon) \eta^*$, and there exists $\rho_2(\epsilon) < 1$ such that $f_s(\eta) < \rho_2(\epsilon) \eta$ as long as $\eta > (1+\epsilon) \eta^*$.
\end{enumerate}
\end{proposition}

We have transformed the problem of computing $\omega_\diamond(Q,s)$ into one of finding the positive fixed point of a one-dimensional function $f_s(\eta)$. Property 6) of Proposition \ref{pro:fix_fs} states that we could start with any $\eta_0$ and use the iteration
\begin{eqnarray}
  \eta_{t+1} = f_s(\eta_{t}), t = 0, 1, \cdots
\end{eqnarray}
to find the positive fixed point $\eta^*$. In addition, if we start from two initial points, one less than $\eta^*$ and one greater than $\eta^*$, then the gap between the generated sequences indicates how close we are to the fixed point $\eta^*$.

Property 5) suggests finding $\eta^*$ by bisection search. Suppose we have an interval $(\eta_\mathrm{L}, \eta_\mathrm{U})$ that includes $\eta^*$. Consider the middle point $\eta_\mathrm{M} = \frac{\eta_\mathrm{L}+\eta_\mathrm{U}}{2}$. If $f_s(\eta_\mathrm{M}) < \eta_\mathrm{M}$, we conclude that $\eta^* < \eta_\mathrm{M}$ and we set $\eta_\mathrm{U} = f(\eta_\mathrm{M})$; if $f_s(\eta_\mathrm{M}) > \eta_\mathrm{M}$, we conclude that $\eta^* > \eta_\mathrm{M}$ and we set $\eta_\mathrm{L} = f(\eta_\mathrm{M})$. We continue this bisection procedure until the interval length $\eta_\mathrm{U}-\eta_\mathrm{L}$ is sufficiently small.

\subsection{Relaxation of the Subproblem}
Unfortunately, except when $n = 1$ and the signal is real, \emph{i.e.,} the real sparse case, it is not easy to compute $f_s(\eta)$ according to \eqref{def:fs}. In the following theorem, we present a relaxation of the subproblem
\begin{eqnarray}\label{eqn:subproblem}
   \max_{\z} \|\z\|_{\b \infty}\text{\ s.t. \ } \|Q\z\|_\diamond \leq 1, {\|\z\|_{\b 1}} \leq s \eta
\end{eqnarray}
by computing an upper bound on $f_s(\eta)$. This proof is similar to that of Proposition \ref{pro:max_sparse} and is given in Appendix \ref{app:pf:relax_subproblem}.
\begin{proposition}\label{pro:relax_subproblem}
When $Q = A$ and $\diamond = 2$, we have
\begin{eqnarray}
  f_s(\eta) &\leq& \max_i \min_{P_i}\max_{j} s\eta\|\delta_{ij}\I_{n}-P_i^TQ_j\|_2 + \|P_i\|_2, and
\end{eqnarray}
when $Q = A^TA$ and $\diamond = \b \infty$, we have
\begin{eqnarray}
  f_s(\eta)  &\leq& \max_i \min_{P_i}\max_{j} s\eta\|\delta_{ij}\I_{n}-P_i^TQ_j\|_2 + \sum_{l = 1}^p\|P_i^l\|_2.
\end{eqnarray}
Here $P_i$ (resp. $Q_j$) is the submatrix of $P$ (resp. $Q$) formed by the $(i-1)n+1$th to $in$th columns (resp. $(j-1)n+1$th to $jn$th columns), and $P_i^l$ is the submatrix of $P$ formed by the $(i-1)n+1$th to $in$th columns and the $(l-1)n+1$th to $ln$th rows.
\end{proposition}

For each $i = 1, \ldots, p$, the optimization problem
\begin{eqnarray}\label{eqn:relax1}
\min_{P_i}\max_{j} s\eta\|\delta_{ij}\I_{n}-P_i^TQ_j\|_2 + \|P_i\|_2
\end{eqnarray}
can be solved using semidefinite programming:
\begin{eqnarray}
&&\min_{P_i, t_0, t_1} s\eta t_0 + t_1 \text{\ s.t. \ } \|\delta_{ij}\I_n - P_i^TQ_j\|_2 \leq t_0, j = 1, \ldots, p, \|P_i\|_2 \leq t_1,\nonumber\\
  &\Leftrightarrow&\min_{P_i, t_0, t_1} s\eta t_0 + t_1 \text{\ s.t. \ }\nonumber\\
  && \left[
     \begin{array}{cc}
       t_0\I_n & \delta_{ij}\I_n - P_i^TQ_j \\
       \delta_{ij}\I_n - Q_j^TP_i & t_0\I_n \\
     \end{array}
   \right] \succeq 0, j = 1, \ldots, p,\nonumber\\
   &&\left[
     \begin{array}{cc}
       t_1 \I_m & P_i \\
       P_i^T & t_1 \I_n \\
     \end{array}
   \right] \succeq 0.
\end{eqnarray}
Similarly, the optimization problem
\begin{eqnarray}\label{eqn:relax2}
  \min_{P_i}\max_{j} s\eta\|\delta_{ij}\I_{n}-P_i^TQ_j\|_2 + \sum_{l = 1}^p\|P_i^l\|_2
\end{eqnarray}
can be solved by the following semidefinite program:
\begin{eqnarray}
&&\min_{P_i, t_0, t_1, \ldots, t_p} s\eta t_0 + \sum_{l=1}^pt_l \text{\ s.t. \ } \|\delta_{ij}\I_n - P_i^TQ_j\|_2 \leq t_0, j = 1, \ldots, p, \|P_i^l\|_2 \leq t_l, l = 1, \ldots, p, \nonumber\\
  &\Leftrightarrow&\min_{P_i, t_0, t_1} s\eta t_0 + \sum_{l=1}^pt_l \text{\ s.t. \ }\nonumber\\
  && \left[
     \begin{array}{cc}
       t_0\I_n & \delta_{ij}\I_n - P_i^TQ_j \\
       \delta_{ij}\I_n - Q_j^TP_i & t_0\I_n \\
     \end{array}
   \right] \succeq 0, j = 1, \ldots, p,\nonumber\\
   &&\left[
     \begin{array}{cc}
       t_l \I_n & P_i^l \\
       P_i^{lT} & t_l \I_n \\
     \end{array}
   \right] \succeq 0, l = 1, \ldots, p.
\end{eqnarray}

These semidefinite programs can also be solved using first-order techniques.

\subsection{Fixed Point Theory for  Computing a Lower Bound on $\omega_\diamond$}\label{sec:computelowerbd}

Although Proposition \ref{pro:relax_subproblem} provides ways to efficiently compute upper bounds on the subproblem \eqref{eqn:subproblem} for fixed $\eta$, it is not obvious whether we could use it to compute an upper bound on the positive fixed point of $f_s(\eta)$, or $1/\omega_\diamond(Q,s)$. We show in this subsection that another iterative procedure can compute such upper bounds.

To this end, we define functions $g_{s,i}(\eta)$ and $g_s(\eta)$ over $[0, \infty)$, parameterized by $s$ for $s\in (1, s_*)$,
\begin{eqnarray}
  g_{s,i}(\eta) &=& \min_{P_i} s\eta  \left(\max_j \|\delta_{ij}\I_n - P_i^T Q_j\|_2\right) + \|P_i\|_2, \text{\ and}\nonumber\\
  g_s(\eta) & = & \max_i g_{s, i}(\eta).
\end{eqnarray}

The following proposition, whose proof is given in Appendix \ref{app:pf:propertygs}, lists some properties of $g_{s,i}(\eta)$ and $g_s(\eta)$.
\begin{proposition}\label{pro:propertygs}
The functions $g_{s,i}(\eta)$ and $g_s(\eta)$ have the following properties:
\begin{enumerate}
  \item $g_{s,i}(\eta)$ and $g_s(\eta)$ are continuous in $\eta$.
  \item $g_{s,i}(\eta)$ and $g_s(\eta)$ are strictly increasing in $\eta$.
  \item $g_{s,i}(\eta)$ is concave for every $i$.
\item $g_s(0) = 0$, $g_s(\eta) \geq s\eta > \eta$ for sufficiently small $\eta > 0$, and there exists $\rho < 1$ such that $g_s(\eta) < \rho\eta$ for sufficiently large $\eta$; the same holds for $g_{s,i}(\eta)$.
    \item $g_{s,i}$ and $g_s(\eta)$ have unique positive fixed points $\eta_i^* = g_{s,i}(\eta_i^*)$ and $\eta^* = g_s(\eta^*)$, respectively, and $\eta^* = \max_i \eta_i^*$.
    \item For $\eta \in (0, \eta^*)$, we have $g_s(\eta) > \eta$, and for $\eta \in (\eta^*, \infty)$, we have $g_s(\eta) < \eta$; the same statement holds also for $g_{s,i}(\eta)$.
  \item For any $\epsilon > 0$, there exists $\rho_1(\epsilon) > 1$ such that $g_s(\eta) > \rho_1(\epsilon) \eta$ as long as $0 < \eta \leq (1-\epsilon) \eta^*$, and there exists $\rho_2(\epsilon) < 1$ such that $g_s(\eta) < \rho_2(\epsilon) \eta$ as long as $\eta > (1+\epsilon) \eta^*$.
\end{enumerate}
\end{proposition}

The same properties in Proposition \ref{pro:propertygs} hold for the functions defined below:
\begin{eqnarray}
  h_{s, i}(\eta) &=& \min_{P_i} s\eta  \left(\max_j \|\delta_{ij}\I_n - P_i^T Q_j\|_2\right) + \sum_{l=1}^p\|P_i^l\|_2, \text{\ and}\\
  h_s(\eta) & = & \max_i h_{s,i}(\eta).
\end{eqnarray}

\begin{proposition}\label{pro:propertyhs}
The functions $h_{s,i}(\eta)$ and $h_s(\eta)$ have the following properties:
\begin{enumerate}
  \item $h_{s,i}(\eta)$ and $h_s(\eta)$ are continuous in $\eta$.
  \item $h_{s,i}(\eta)$ and $h_s(\eta)$ are strictly increasing in $\eta$.
  \item $h_{s,i}(\eta)$ is concave for every $i$.
\item $h_s(0) = 0$, $h_s(\eta) \geq s\eta > \eta$ for sufficiently small $\eta > 0$, and there exists $\rho < 1$ such that $h_s(\eta) < \rho\eta$ for sufficiently large $\eta$; the same holds for $h_{s,i}(\eta)$.
    \item $h_{s,i}$ and $h_s(\eta)$ have unique positive fixed points $\eta_i^* = h_{s,i}(\eta_i^*)$ and $\eta^* = h_s(\eta^*)$, respectively, and $\eta^* = \max_i \eta_i^*$.
    \item For $\eta \in (0, \eta^*)$, we have $h_s(\eta) > \eta$, and for $\eta \in (\eta^*, \infty)$, we have $h_s(\eta) < \eta$; the same statement holds also for $h_{s,i}(\eta)$.
  \item For any $\epsilon > 0$, there exists $\rho_1(\epsilon) > 1$ such that $h_s(\eta) > \rho_1(\epsilon) \eta$ as long as $0 < \eta \leq (1-\epsilon) \eta^*$, and there exists $\rho_2(\epsilon) < 1$ such that $h_s(\eta) < \rho_2(\epsilon) \eta$ as long as $\eta > (1+\epsilon) \eta^*$.
\end{enumerate}
\end{proposition}

An immediate consequence of Propositions \ref{pro:propertygs} and \ref{pro:propertyhs} is the following:
\begin{theorem}
Suppose $\eta^*$ is the unique fixed point of $g_s(\eta)$ ($h_s(\eta)$, resp.), then we have
\begin{eqnarray}
\eta^* \geq \frac{1}{\omega_2(A,s)} \left(\frac{1}{\omega_{\b \infty}(A^TA,s)}, \text{\ resp.\ }\right).
\end{eqnarray}
\end{theorem}

Proposition \ref{pro:propertygs} implies three ways to compute the fixed point $\eta^*$ for $g_s(\eta)$. The same discussion is also valid for $h_s(\eta)$.
\begin{enumerate}
  \item \textbf{Naive Fixed Point Iteration:} Property 7) of Proposition \ref{pro:propertygs} suggests that the fixed point iteration
      \begin{eqnarray}
        \eta_{t+1} &=& g_s(\eta_t), t = 0, 1, \ldots
      \end{eqnarray}
      starting from any initial point $\eta_0$ converges to $\eta^*$, no matter whether $\eta_0 < \eta^*$ or $\eta_0 > \eta^*$. The algorithm can be made more efficient in the case $\eta_0 < \eta^*$. More specifically, since $g_s(\eta) = \max_i g_{s,i}(\eta)$, at each fixed point iteration, we set $\eta_{t+1}$ to be the first $g_{s,i}(\eta_t)$ that is greater than $\eta_t + \epsilon$, with $\epsilon$ being some tolerance parameter. If for all $i$, $g_{s,i}(\eta_t) < \eta_t + \epsilon$, then $g_s(\eta_t) = \max_i g_{s,i}(\eta_t) < \eta_t + \epsilon$, which indicates that the optimal function value can not be improved greatly and the algorithm should terminate. In most cases, to get $\eta_{t+1}$, we need to solve only one optimization problem, $\min_{P_i} s\eta  \left(\max_j \|\delta_{ij}\I_n - P_i^T Q_j\|_2\right) + \|P_i\|_2$, instead of solving for $p$. This is in contrast to the case where $\eta_0 > \eta^*$, because in the latter case we must compute all $g_{s,i}(\eta_t)$ to update $\eta_{t+1} = \max_i g_{s,i}(\eta_t)$. An update based on a single $g_{s,i}(\eta_t)$ might generate a value smaller than $\eta^*$.

      The naive fixed point iteration has two major disadvantages. First, the stopping criterion based on successive improvement is not accurate as it does not reflect the gap between $\eta_t$ and $\eta^*$. This disadvantage can be remedied by starting from both below and above $\eta^*$. The distance between corresponding terms in the two generated sequences is an indication of the gap to the fixed point $\eta^*$. However, the resulting algorithm is very slow, especially when updating $\eta_{t+1}$ from above $\eta^*$. Second, the iteration process is slow, especially when close to the fixed point $\eta^*$, because $\rho_1(\epsilon)$ and $\rho_2(\epsilon)$ in 7) of Proposition \ref{pro:propertygs} are close to 1.
  \item \textbf{Bisection:} The bisection approach is motivated by property 6) of Proposition \ref{pro:propertygs}. Starting from an initial interval $(\eta_\mathrm{L}, \eta_\mathrm{U})$ that contains $\eta^*$, we compute $g_s(\eta_\mathrm{M})$ with $\eta_\mathrm{M} = (\eta_{\mathrm{L}} + \eta_\mathrm{U})/2$. As a consequence of property 6), $g_s(\eta_\mathrm{M}) > \eta_\mathrm{M}$ implies $g_s(\eta_\mathrm{M}) < \eta^*$, and we set $\eta_\mathrm{L} = g_s(\eta_\mathrm{M})$; $g_s(\eta_\mathrm{M}) < \eta_\mathrm{M}$ implies $g_s(\eta_\mathrm{M}) > \eta^*$, and we set $\eta_\mathrm{U} = g_s(\eta_\mathrm{M})$. The bisection process can also be accelerated by setting $\eta_\mathrm{L} = g_{s,i}(\eta_\mathrm{M})$ for the first $g_{s,i}(\eta_\mathrm{M})$ greater than $\eta_\mathrm{M}$. The convergence of the bisection approach is much faster than the naive fixed point iteration because each iteration reduces the interval length at least by half. In addition, half the length of the interval is an upper bound on the gap between $\eta_{\mathrm{M}}$ and $\eta^*$, resulting an accurate stopping criterion. However, if the initial $\eta_{\mathrm{U}}$ is too much larger than $\eta^*$, the majority of $g_s(\eta_{\mathrm{M}})$ would turn out to be less than $\eta^*$. The verification of $g_s(\eta_{\mathrm{M}}) < \eta_{\mathrm{M}}$ requires solving $p$ semidefinite programs, greatly degrading the algorithm's performance.
  \item \textbf{Fixed Point Iteration $+$ Bisection:} The third approach combines the advantages of the bisection method and the fixed point iteration method, at the level of $g_{s,i}(\eta)$. This method relies on the representations $g_s(\eta) = \max_i g_{s,i}(\eta)$ and $\eta^* = \max_i \eta_i^*$.

      Starting from an initial interval $(\eta_{\mathrm{L}0}, \eta_\mathrm{U})$ and the index set $\sI_0 = \{1,\ldots, p\}$, we pick any $i_0 \in \sI_0$ and use the (accelerated) bisection method with starting interval $(\eta_{\mathrm{L}0}, \eta_{\mathrm{U}})$ to find the positive fixed point $\eta_{i_0}^*$ of $g_{s,i_0}(\eta)$. For any $i \in \sI_0/i_0$, $g_{s,i}(\eta_{i_0}^*) \leq \eta_{i_0}^*$ implies that the fixed point $\eta_i^*$ of $g_{s,i}(\eta)$ is less than or equal to $\eta_{i_0}^*$ according to the continuity of $g_{s,i}(\eta)$ and the uniqueness of its positive fixed point. As a consequence, we remove this $i$ from the index set $\sI_0$. We denote $\sI_1$ as the index set after all such $i$s are removed, \emph{i.e.,} $\sI_1 = \sI_0/\{i: g_{s,i}(\eta_{i_0}^*) \leq \eta_{i_0}^*\}$. We also set $\eta_{L1} = \eta_{i_0}^*$ as $\eta^* \geq \eta_{i_0}^*$. Next we test the $i_1 \in \sI_1$ with the \emph{largest} $g_{s,i}(\eta_{i_0}^*)$ and construct $\sI_2$ and $\eta_{\mathrm{L}2}$ in a similar manner. We repeat the process until the index set $\sI_t$ is empty. The $\eta_i^*$ found at the last step is the maximal $\eta_i^*$, which is equal to $\eta^*$.
\end{enumerate}

\section{Preliminary Numerical Experiments}\label{sec:numerical}
In this section, we present preliminary numerical results that assess the performance of the algorithms for verifying $\omega_2(A,s) > 0$ and computing $\omega_2(A,s)$. We also compare the error bounds based on $\omega_2(A,s)$ with the bounds based on the block RIP \cite{Eldar2009BlockSparse}. The involved semidefinite programs are solved using CVX.

We test the algorithms on Gaussian random matrices. The entries of Gaussian matrices are randomly generated from the standard Gaussian distribution. All $m\times np$ matrices are normalized to have columns of unit length.

We first present the values of $s_*$ computed by \eqref{eqn:max_sparse}, and the values of $k_*$ ($k_* = \lfloor s_*/2 \rfloor$), and compare them with the corresponding quantities when $A$ is used as the sensing matrix for sparsity recovery without knowing the block-sparsity structure. The quantities in the latter case are computed using the algorithms developed in \cite{tang2011cmsv, tang2011linf}. We note in Table \ref{tbl:sparselevel} that for the same sensing matrix $A$, both $s_*$ and $k_*$ are smaller when the block-sparse structure is taken into account than when it is not taken into account. However, we need to keep in mind that the true sparsity level in the block-sparse model is $nk$, where $k$ is the block sparsity level. The $nk_*$ in the fourth column for the block-sparse model is indeed much greater than the $k_*$ in the sixth column for the sparse model, implying that exploiting the block-sparsity structure is advantageous.

\begin{table}[h!t]
\caption{Comparison of the sparsity level bounds on block-sparsity recovery and sparsity recovery for a Gaussian sensing matrix $A \in \R^{m\times np}$ with $n = 4, p = 60$.}
\begin{center}
\begin{tabular}{||l||l|l|l||l|l||}
\hline
\multirow{2}{*}{$m$} & \multicolumn{3}{c||}{Block Sparse Model} & \multicolumn{2}{c||}{Sparse Model}\\
\cline{2-3}\cline{4-6}
& $s_*$ & $k_*$ & $nk_*$ & $s_*$ & $k_*$\\
\hline
\hline
72 & 3.96 & 1 &  4 & 6.12 & 3\\
\hline
 96 & 4.87 & 2 &  8 & 7.55 & 3\\
\hline
120 & 5.94 & 2 &  8 & 9.54 & 4\\
\hline
144 & 7.14 & 3 & 12 & 11.96 & 5\\
\hline
168 & 8.60 & 4 & 16 & 14.66 & 7\\
\hline
192 & 11.02 & 5 & 20 & 18.41 & 9\\
\hline
\end{tabular}\label{tbl:sparselevel}
\end{center}
\end{table}

In the next set of experiments, we compare the computation times for the three implementation methods discussed at the end of Section \ref{sec:computelowerbd}. The Gaussian matrix $A$ is of size $72\times 120$ with $n = 3$ and $p = 40$. The tolerance parameter is $10^{-5}$. The initial $\eta$ value for the naive fixed point iteration is $0.1$. The initial lower bound $\eta_\mathrm{L}$ and upper bound $\eta_{\mathrm{U}}$ are set as $0.1$ and $10$, respectively. All three implementations yield $\eta^* = 0.7034$. The CPU times for the three methods are $393$ seconds, $1309$ seconds, and $265$ seconds. Therefore, the fixed point iteration + bisection gives the most efficient implementation in general. The bisection search method is slow in this case because the initial $\eta_U$ is too much larger than $\eta^*$.

In the last experiment, we compare our recovery error bounds on the BS-BP based on $\omega_2(A,s)$ with those based on the block RIP. Recall that from Corollary \ref{cor:l2errorbound}, we have for the BS-BP
\begin{eqnarray}\label{eqn:bp_omega_bd}
  \|\hx - \x\|_2 \leq \frac{2\sqrt{2k}}{\omega_2(A,2k)}\varepsilon.
\end{eqnarray}
For comparison, the block RIP  bounds is
\begin{eqnarray}\label{eqn:bp_rip_bd}
\|\hx-\x\|_2 \leq \frac{4\sqrt{1+\delta_{2k}(A)}}{1-(1+\sqrt{2})\delta_{2k}(A)}\varepsilon,
\end{eqnarray}
assuming the block RIP $\delta_{2k}(A) < \sqrt{2}-1$ \cite{Eldar2009BlockSparse}. Without loss of generality, we set $\varepsilon = 1$.

The block RIP is computed using Monte Carlo simulations. More explicitly, for $\delta_{2k}(A)$, we randomly take $1000$ sub-matrices of $A \in \R^{m\times np}$ of size $m\times 2nk$ with a pattern determined by the block-sparsity structure, compute the maximal and minimal singular values $\sigma_1$ and $\sigma_{2k}$, and approximate $\delta_{2k}(A)$ using the maximum of $\max(\sigma_1^2 - 1, 1 - \sigma_{2k}^2)$ among all sampled sub-matrices. Obviously, the approximated block RIP is always smaller than or equal to the exact block RIP. As a consequence, the performance bounds based on the exact block RIP are \emph{worse} than those based on the approximated block RIP. Therefore, in cases where our $\omega_2(A,2k)$ based bounds are better (tighter, smaller) than the approximated block RIP bounds, they are even better than the exact block RIP bounds.

In Table \ref{tbl:Gaussian_omega_bric}, we present the values of $\omega_2(A,2k)$ and $\delta_{2k}(A)$ computed for a Gaussian matrix $A \in \R^{m\times np}$ with $n = 4$ and $p = 60$. The corresponding $s_*$ and $k_*$ for different $m$ are also included in the table. We note that in all the considered cases, $\delta_{2k}(A) > \sqrt{2} - 1$, and the block RIP based bound \eqref{eqn:bp_rip_bd} does not apply at all. In contrast, the $\omega_2$ based bound \eqref{eqn:bp_omega_bd} is valid as long as $k \leq k_*$. In Table \eqref{tbl:omega2_bd} we show the $\omega_2$ based bound \eqref{eqn:bp_omega_bd}.

\begin{table}[h!t]
\caption{$\omega_2(A,2k)$ and $\delta_{2k}(A)$ computed for a Gaussian matrix $A \in \R^{m\times np}$ with $n = 4$ and $p = 60$.}
\begin{center}
\hskip -.3 cm

\begin{tabular}{||l|l||l|l|l|l|l|l|}
\cline{2-8}
\multicolumn{1}{l|}{\multirow{2}{*}{}} & $m$ & 72 & 96 & 120 & 144 & 168 & 192\\
\cline{2-8}
\multicolumn{1}{l|}{}& $s_*$ & 3.88  &  4.78  &  5.89  &  7.02  &  8.30 &  10.80\\
\hline
$k$ & $k_*$ & 1  &   2  &   2  &   3  &   4  &   5\\
\hline\hline

\multirow{2}{*}{1} & $\omega_2(A,2k)$ & 0.45 &   0.53  &  0.57  &  0.62 &   0.65  &  0.67\\
& $\delta_{2k}(A)$ & 0.90 &    0.79 &    0.66 &    0.58 &    0.55 &    0.51\\
\hline

\multirow{2}{*}{2} & $\omega_2(A,2k)$ & \multirow{2}{*}{-----} &   0.13 &    0.25 &    0.33 &    0.39 &    0.43\\
& $\delta_{2k}(A)$ & & 1.08 &    0.98 &    0.96 &    0.84 &    0.75\\
\hline

\multirow{2}{*}{3} & $\omega_2(A,2k)$ & \multirow{2}{*}{-----}  & \multirow{2}{*}{-----}   & \multirow{2}{*}{-----}   &  0.11 &    0.18 &    0.25\\
& $\delta_{2k}(A)$ & & & &1.12 &    1.01 &    0.93\\
\hline

\multirow{2}{*}{4} & $\omega_2(A,2k)$ & \multirow{2}{*}{-----}  & \multirow{2}{*}{-----}   & \multirow{2}{*}{-----}  & \multirow{2}{*}{-----}   &    0.02 &    0.12\\
& $\delta_{2k}(A)$ & & & & &    1.26 &    1.07\\
\hline

\multirow{2}{*}{5} & $\omega_2(A,2k)$ & \multirow{2}{*}{-----}   & \multirow{2}{*}{-----}   & \multirow{2}{*}{-----}   & \multirow{2}{*}{-----}   & \multirow{2}{*}{-----}  &    0.03\\
& $\delta_{2k}(A)$ & & & & &  &    1.28\\
\hline
\end{tabular}\label{tbl:Gaussian_omega_bric}
\end{center}
\end{table}

\begin{table}[h!t]
\caption{The $\omega_2(A,2k)$ based bounds on the $\ell_2$ norms of the errors of the BS-BP for the Gaussian matrix in Table \ref{tbl:Gaussian_omega_bric}.}
\begin{center}
\hskip -.3 cm

\begin{tabular}{||l|l||l|l|l|l|l|l|}
\cline{2-8}
\multicolumn{1}{l|}{\multirow{2}{*}{}} & $m$ & 72 & 96 & 120 & 144 & 168 & 192\\
\cline{2-8}
\multicolumn{1}{l|}{}& $s_*$ & 3.88  &  4.78  &  5.89  &  7.02  &  8.30 &  10.80\\
\hline
$k$ & $k_*$ & 1  &   2  &   2  &   3  &   4  &   5\\
\hline\hline

1 & $\omega_2$ bound & 6.22 &   13.01 &    9.89 &    6.50 &   11.52 &    9.50\\
\hline

2 & $\omega_2$ bound & ----- &   58.56 &   25.37 &   14.64 &    7.30 &   16.26\\
\hline

3 & $\omega_2$ bound & -----  & -----   &-----   &  53.54 &   21.63 &   30.27\\
\hline

4 & $\omega_2$ bound &-----   &-----    & -----  & -----  &    236.74 &   23.25\\
\hline

5 & $\omega_2$ bound &-----   &-----    & -----  & -----  &  -----  &    127.59\\
\hline
\end{tabular}\label{tbl:omega2_bd}
\end{center}
\end{table}

\section{Conclusions}\label{sec:conclusions}
In this paper, we analyzed the performance of convex block-sparse signal recovery algorithms using the block-$\ell_{\b \infty}$ norm of the errors as a performance criterion. We expressed other popular performance criteria in terms of the block-$\ell_{\b \infty}$ norm. A family of goodness measures of the sensing matrices was defined using optimization procedures. We used these goodness measures to derive upper bounds on the block-$\ell_{\b \infty}$ norms of the reconstruction errors for the Block-Sparse Basis Pursuit, the Block-Sparse Dantzig Selector, and the Block-Sparse LASSO estimator. Efficient algorithms based on fixed point iteration, bisection, and semidefinite programming were implemented to solve the optimization procedures defining the goodness measures. We expect that these goodness measures will be useful in comparing different sensing systems and recovery algorithms, as well as in designing optimal sensing matrices. In future work, we will use these computable performance bounds to optimally design transmitting waveforms for compressive sensing radar.
\section{Appendix: Proofs}
\subsection{Proof of Proposition \ref{pro:errorcharacteristics}}\label{app:pf:errorcharacteristics}

\begin{proof} Suppose $S = \bsupp(\x)$ and $|S| = \|\x\|_{\b 0} = k$. Define the error vector $\h = \hx - \x$.

We first prove the proposition for the BS-BP and the BS-DS. The fact that $\|\hx\|_{\b 1} = \|\x + \h\|_{\b 1}$ is the minimum among all $\z$s satisfying the constraints in \eqref{bp} and \eqref{ds}, together with the fact that the true signal $\x$ satisfies the constraints as required by the conditions imposed on the noise in Proposition \ref{pro:errorcharacteristics}, implies that $\|\h_{S^c}\|_{\b 1}$ cannot be very large. To see this, note that
\begin{eqnarray}\label{x_min}
  \|\x\|_{\b 1} &\geq& \|\x + \h\|_{\b 1}\nonumber\\
  & = & \sum_{i\in S}\|\x_i + \h_i\|_2 + \sum_{i\in S^c}\|\x_i + \h_i\|_2 \nonumber\\
  &\geq & \sum_{i\in S}\|\x_i\|_2 - \sum_{i\in S}\|\h_i\|_2 + \sum_{i\in S^c}\|\h_i\|_2 \nonumber\\
  &=& \|\x_S\|_{\b 1} - \|\h_S\|_{\b 1} + \|\h_{S^c}\|_{\b 1}\nonumber\\
  & = & \|\x\|_{\b 1} - \|\h_S\|_{\b 1} + \|\h_{S^c}\|_{\b 1}.
\end{eqnarray}
Therefore, we obtain $\|\h_{S}\|_{\b 1} \geq \|\h_{S^c}\|_{\b 1}$, which leads to
\begin{eqnarray}\label{h1h2}
2\|\h_S\|_{\b 1}  \geq \|\h_S\|_{\b 1} + \|\h_{S^c}\|_{\b 1} =  \|\h\|_{\b 1}.
\end{eqnarray}

We now turn to the BS-LASSO \eqref{lasso}. Since the noise $\w$ satisfies $\|A^T\w\|_{\b \infty}\leq \kappa\mu$ for some $\kappa \in (0, 1)$, and $\hx$ is the minimizer of \eqref{lasso}, we have
\begin{eqnarray*}
\frac{1}{2}\|A\hx - \y\|_2^2 + \mu \|\hx\|_{\b 1} \leq \frac{1}{2} \|A\x - \y\|_2^2 + \mu\|\x\|_{\b 1}.
\end{eqnarray*}
Consequently, substituting $\y = A\x + \w$ yields
\begin{eqnarray*}
\mu \|\hx\|_{\b 1} &\leq& \frac{1}{2} \|A\x - \y\|_2^2 - \frac{1}{2}\|A\hx - \y\|_2^2 + \mu \|\x\|_{\b 1} \nn\\
&=& \frac{1}{2}\|\w\|_2^2 - \frac{1}{2}\|A(\hx-\x) - \w\|_2^2 + \mu\|\x\|_{\b 1}\nn\\
& = & \frac{1}{2}\|\w\|_2^2 - \frac{1}{2}\|A(\hx-\x)\|_2^2\nn\\
  && \ \ \ + \left<A(\hx-\x), \w\right> - \frac{1}{2}\|\w\|_2^2 + \mu\|\x\|_{\b 1}\nn\\
&\leq & \left<A(\hx-\x), \w\right> + \mu\|\x\|_{\b 1}\nn\\
& = & \left<\hx - \x, A^T\w\right> + \mu\|\x\|_{\b 1}.
\end{eqnarray*}
Using the Cauchy-Swcharz type inequality, we get
\begin{eqnarray*}
\mu \|\hx\|_{\b 1} &\leq & \|\hx-\x\|_{\b 1}\|A^T\w\|_{\b \infty} + \mu\|\x\|_{\b 1}\nn\\
& = & \kappa \mu \|\h\|_{\b 1} + \mu\|\x\|_{\b 1},
\end{eqnarray*}
which leads to
\begin{eqnarray*}
  \|\hx\|_{\b 1} &\leq & \kappa\|\h\|_{\b 1} + \|\x\|_{\b 1}.
\end{eqnarray*}
Therefore, similar to the argument in \eqref{x_min}, we have
\begin{eqnarray*}
 &&\|\x\|_{\b 1} \nn\\
 &\geq& \|\hx\|_{\b 1} - \kappa\|\h\|_{\b 1}\nn\\
 & = & \|\x + \h_{S^c} + \h_S\|_{\b 1}- \kappa\left(\|\h_{S^c} + \h_S\|_{\b 1}\right) \nonumber\\
  &\geq& \|\x + \h_{S^c} \|_{\b 1} - \|\h_S\|_{\b 1} - \kappa\left(\|\h_{S^c}\|_{\b 1} + \|\h_S\|_{\b 1}\right) \nonumber\\
  & = & \|\x\|_{\b 1} + (1-\kappa)\|\h_{S^c}\|_{\b 1} - (1+\kappa)\|\h_S\|_{\b 1},
\end{eqnarray*}
where $S = \bsupp(\x)$.
Consequently, we have
\begin{eqnarray*}
  \|\h_{S}\|_{\b 1} &\geq& \frac{1-\kappa}{1+\kappa}\|\h_{S^c}\|_{\b 1}.
\end{eqnarray*}
Therefore, similar to \eqref{h1h2}, we obtain
\begin{eqnarray}\label{lassoh1h2}
\frac{2}{1-\kappa}\|\h_S\|_{\b 1} &=& \frac{1+\kappa}{1-\kappa}\|\h_S\|_{\b 1} + \frac{1-\kappa}{1-\kappa} \|\h_{S}\|_{\b 1}\nn\\
&\geq& \frac{1+\kappa}{1-\kappa}\frac{1-\kappa}{1+\kappa}\|\h_{S^c}\|_{\b 1} + \frac{1-\kappa}{1-\kappa} \|\h_{S}\|_{\b 1}\nn\\
& = & \|\h\|_{\b 1}. 
\end{eqnarray}
\qed
\end{proof}

\subsection{Proof of Corollary \ref{cor:connections}}\label{app:pf:connections}
\begin{proof}
Suppose $S = \bsupp(\x)$. According to Proposition \ref{pro:errorcharacteristics}, we have
\begin{eqnarray}
  \|\h\|_{\b 1} \leq c \|\h_S\|_{\b 1} \leq  c k \|\h\|_{\b \infty}.
\end{eqnarray}

To prove \eqref{eqn:l2linf}, we note
\begin{eqnarray}
  \frac{\|\h\|_2^2}{\|\h\|_{\b \infty}^2} &=& \sum_{i=1}^p\left(\frac{\|\h_i\|_2}{\|\h\|_{\b \infty}}\right)^2\nonumber\\
  &\leq &  \sum_{i=1}^p\left(\frac{\|\h_i\|_2}{\|\h\|_{\b \infty}}\right)\nonumber\\
  &= & \frac{\|\h\|_{\b 1}}{\|\h\|_{\b \infty}}\nonumber\\
  &\leq & ck.
\end{eqnarray}
For the first inequality, we have used $\frac{\|\h_i\|_2}{\|\h\|_{\b \infty}}\leq 1$ and $a^2 \leq a$ for $ a \in [0,1]$.

For the last assertion, note that if $\|\h\|_{\b \infty} \leq \beta/2$, then we have for $i \in S$,
\begin{eqnarray}
  \|\hx_i\|_2 = \|\x_i + \h_i\|_2 \geq \|\x_i\|_2 - \|\h_i\|_2 > \beta - \beta/2 = \beta/2;
\end{eqnarray}
and for $i \notin S$,
\begin{eqnarray}
  \|\hx_i\|_2 = \|\x_i + \h_i\|_2 = \|\h_i\|_2 < \beta/2.
\end{eqnarray}
\qed
\end{proof}

\subsection{Proof of Theorem \ref{thm:errorbound}}\label{app:pf:errorbound}
\begin{proof}
Observe that for the BS-BP
\begin{eqnarray}
\|A(\hx - \x)\|_2 &\leq& \|\y - A\hx\|_2 + \|\y - A\x\|_2\nn\\
&\leq& \varepsilon + \|A\w\|_2\nn\\
&\leq& 2\varepsilon,
\end{eqnarray}
and similarly,
\begin{eqnarray}
  \|A^TA(\hx - \x)\|_{\b \infty}\leq 2\mu
\end{eqnarray}
for the BS-DS, and
\begin{eqnarray}
\|A^TA(\hx - \x)\|_{\b \infty}&\leq& (1+\kappa)\mu
\end{eqnarray}
for the BS-LASSO. Here for the BS-LASSO, we have used
\begin{eqnarray}
  \|A^T(A\hx -\y)\|_{\b \infty} \leq \mu,
\end{eqnarray}
a consequence of the optimality condition
\begin{eqnarray}
  A^T(A\hx - \y) \in \mu \partial \|\hx\|_{\b 1}
\end{eqnarray}
and the fact that the $i$th block of any subgradient in $\partial \|\hx\|_{\b 1}$ is $\hx_i/\|\hx_i\|_2$ if $\hx_i \neq 0$ and is $\g$ otherwise for some $\|\g\|_2 \leq 1$.

The conclusions of Theorem \ref{thm:errorbound} follow from equation \eqref{eqn:l1linf} and Definition \ref{def:linfcmsv}. \qed
\end{proof}

\subsection{Proof of Lemma \ref{lm:omega_rho}}\label{app:pf:lm_omega_rho}
\begin{proof}
For any $\z$ such that $\|\z\|_{\b \infty}= 1$ and $\|\z\|_{\b 1} \leq s$, we have
\begin{eqnarray}\label{eqn:quad_binf}
  \z A^TA\z & = &\left<\z, A^TA\z\right>\nn\\
   &\leq & \|\z\|_{\b 1}\|A^TA\z\|_{\b \infty}\nn\\
   &\leq & s\|A^TA\z\|_{\b \infty}.
\end{eqnarray}
Taking the minima of both sides of \eqref{eqn:quad_binf} over $\{\z: \|\z\|_{\b \infty}= 1, \|\z\|_1 \leq s\}$ yields
\begin{eqnarray}
  \omega_2^2(A,s) &\leq& s \omega_{\b \infty}(A^TA,s).
\end{eqnarray}

For the other inequality, note that $\|\z\|_{\b 1}/\|\z\|_{\b \infty}\leq s$ implies $\|\z\|_{\b 1} \leq s\|\z\|_{\b \infty}\leq s \|\z\|_2$, or equivalently,
\begin{eqnarray}
  \{\z: \|\z\|_{\b 1}/\|\z\|_{\b \infty}\leq s \} \subseteqq  \{\z: \|\z\|_{\b 1}/\|\z\|_2 \leq s \}.
\end{eqnarray}
As a consequence, we have
\begin{eqnarray}
  \omega_2(A,s) &=& \min_{\|\z\|_{\b 1}/\|\z\|_{\b \infty}\leq s}\frac{\|A\z\|_2}{\|\z\|_2} \frac{\|\z\|_2}{\|\z\|_{\b \infty}}\nn\\
  &\geq &   \min_{\|\z\|_{\b 1}/\|\z\|_{\b \infty}\leq s}\frac{\|A\z\|_2}{\|\z\|_2}\nn\\
  &\geq &   \min_{\|\z\|_{\b 1}/\|\z\|_2 \leq s}\frac{\|A\z\|_2}{\|\z\|_2}\nn\\
  & = & \rho_{s^2}(A),
\end{eqnarray}
where the first inequality is due to $\|\z\|_2 \geq \|\z\|_{\b \infty}$, and the second inequality is because the minimization is taken over a larger set. \qed
\end{proof}

\subsection{Proof of Theorem \ref{thm:randomcmsv}}\label{app:pf:randomcmsv}
We first need some definitions. Suppose $X$ is a scalar random variable, the Orlicz $\psi_2$ norm of $X$ is defined as
\begin{eqnarray}
\|X\|_{\psi_2} = \inf \left\{t > 0: \E \exp\left(\frac{|X|^2}{t^2}\right) \leq 2\right\}.
\end{eqnarray}
The $\psi_2$ norm is closely related to the subgaussian constant $L$. We actually can equivalently define that a random vector $\bs X \in \R^{np}$ is isotropic and subgaussian with constant $L$ if $\E|\left<\bs X, \u\right>|^2 = \|\u\|_2^2$ and $\|\left<\bs X, \u\right>\|_{\psi_2} \leq L \|\u\|_2$ for any $\u \in \R^{np}$ \cite{tang2011cmsv}.

We use the notation $\ell_*(\H) = \E\ \sup_{\u \in \H} \left<\g, \u\right>$ with $\g \sim \N(0, \I_{np})$, \emph{i.e.}, a vector of independent zero-mean unit-variance Gaussians, to denote the Gaussian width of any set $\H \subset \R^{np}$.

We now cite a result on the behavior of empirical processes \cite{tang2011cmsv, mendelson2007subgaussian}:
\begin{theorem}\hskip -0.1cm\emph{\cite{tang2011cmsv, mendelson2007subgaussian}}\label{thm:empirical}
Let $\{\a, \a_i, i = 1,\ldots,m\} \subset \R^{np}$ be \emph{i.i.d.} isotropic and subgaussian random vectors.  $\H$ is a subset of the unit sphere of $\R^{np}$, and $\sF = \{f_{\u}(\cdot) = \left<\u, \cdot\right>: \u \in \H\}$. Suppose $\mathrm{diam}(\sF, \|\cdot\|_{\psi_2}) \df \max_{f, g \in \sF}\|f-g\|_{\psi_2} = \alpha$. Then there exist absolute constants $c_1, c_2, c_3$ such that for any $\epsilon > 0$ and $m \geq 1$ satisfying
\begin{eqnarray}
  m \geq c_1 \frac{\alpha^2 \ell_*^2(\H)}{\epsilon^2},
\end{eqnarray}
with probability at least $1 - \exp(-c_2\epsilon^2 m/\alpha^4)$,
\begin{eqnarray}
  \sup_{f\in \sF}\left|\frac{1}{m} \sum_{k=1}^m f^2(\a_k) - \E f^2(\a)\right| \leq \epsilon.
\end{eqnarray}
Furthermore, if $\sF$ is symmetric, we have
\begin{eqnarray}
\hskip -1cm &&\E \sup_{f \in \sF} \left|\frac{1}{m} \sum_{k=1}^m f^2(\a_k) - \E f^2(\a)\right| \nn\\
\hskip -1cm &&\leq c_3 \max\left\{\alpha \frac{\ell_*(\H)}{\sqrt{m}}, \frac{\ell_*^2(\H)}{m}\right\}.
\end{eqnarray}
\end{theorem}

With these preparations, we proceed to the proof of Theorem \ref{thm:randomcmsv}: 
\begin{proof}[Proof of Theorem \ref{thm:randomcmsv}] We apply Theorem \ref{thm:empirical} to estimate the block $\ell_1$-CMSV. According to the assumptions of Theorem \ref{thm:randomcmsv}, the rows $\{\a_i\}_{i=1}^m$ of $\sqrt{m} A$ are \emph{i.i.d.} isotropic and subguassian random vectors with constant $L$. Consider the set $\H = \{\u \in \R^{np}: \|\u\|_2^2 = 1, \|\u\|_{\b 1}^2 \leq s\}$ and the function set $\sF = \{f_{\u}(\cdot) = \left<\u, \cdot\right>: \u \in \H\}$. Clearly both $\H$ and $\sF$ are symmetric. The diameter of $\sF$ satisfies
\begin{eqnarray}
\alpha &=& \mathrm{diam}(\sF, \|\cdot\|_{\psi_2})\nn\\
& \leq& 2 \sup_{\u\in \H} \| \left<\u,\a\right>\|_{\psi_2} = 2L.
\end{eqnarray}

Since $\E f^2(\a) = \E\left<\u, \a\right>^2 = \|\u\|_2^2 = 1$ when $\a$ follows the same distribution as $\{\a_i\}_{i=1}^m$, we observe that 
for any $\epsilon \in (0,1)$
\begin{eqnarray}
  \rho_s(A)^2 = \min_{\u: \u \in \H} \u^TA^TA\u < (1-\epsilon)^2 < (1-\epsilon)
\end{eqnarray}
is a consequence of
\begin{eqnarray}\label{eqn:supform}
&&\sup_{\u \in \H}\left|\frac{1}{m}\u^T(\sqrt{m}A)^T(\sqrt{m}A)\u - 1\right|\nn\\
&=&  \sup_{\u \in \H}\left|\frac{1}{m}\sum_{i=1}^m f^2_{\u}(\a_i) - \E f^2_{\u}(\a)\right| \leq \epsilon.
\end{eqnarray}

In view of Theorem \ref{thm:empirical}, the key is to estimate the Gaussian width $\ell_*(\H)$
\begin{eqnarray}\label{eqn:estimate_gwidth}
\ell_*(\H) & =  & \E\ \sup_{\u: \|\u\|_2 = 1, \|\u\|_{\b 1}^2 \leq s} \left<\g, \u\right> \nn\\
&\leq& \E\|\u\|_{\b 1}\|\g\|_{\b \infty}\nn\\
&\leq &\sqrt{s}\ \E\|\g\|_{\b \infty}.
\end{eqnarray}
The quantity $\E \|\g\|_{\b \infty}$ can be bounded using Slepian's inequality \cite[Section 3.3]{ledoux1991probability}. We rearrange the vector $\g \in \R^{np}$ into a matrix $G \in \R^{n\times p}$ such that the vectorization $\mathrm{vec}(G) = \g$. Clearly, we have $\|\g\|_{\b \infty} = \|G\|_{1,2}$, where $\|\cdot\|_{1,2}$ denotes the matrix norm as an operator from $(\R^n, \|\cdot\|_{\ell_1})$ to $(\R^p, \|\cdot\|_{\ell_2})$. Recognizing $\|G\|_{1,2} = \max_{\v \in S^{n-1}, \w \in T^{p-1}}\left<G\v, \w\right>$, we define the Gaussian process $X_{\v, \w} = \left<G\v, \w\right>$ indexed by $(\v, \w) \in S^{n-1}\times T^{p-1}$. Here $S^{n-1} = \{\v\in \R^n: \|\v\|_2 = 1\}$ and $T^{p-1} = \{\w \in \R^p: \|\w\|_1 = 1\}$. We compare $X_{\v, \w}$ with another Gaussian process $Y_{\v, \w} = \left<\bs \xi, \v\right> + \left<\bs \zeta, \w\right>, (\v, \w)\in S^{n-1}\times T^{p-1}$, where $\bs \xi \sim \N(0, \I_n)$ and $\bs \zeta \sim \N(0, \I_p)$. The Gaussian processes $X_{\v, \w}$ and $Y_{\v, \w}$ satisfy the conditions for Slepian's inequality (See the proof of \cite[Theorem 32, page 23]{vershynin2011randommatrices}). Therefore, we have
\begin{eqnarray}\label{eqn:estimate_ginf}
  \E \|\g\|_{\b \infty} &=& \E \max_{(\v,\w) \in S^{n-1}\times T^{p-1}} X_{\v, \w} \leq \E \max_{(\v,\w) \in S^{n-1}\times T^{p-1}} Y_{\v, \w}\nn\\
  & = & \E \max_{\v \in S^{n-1}} \left<\bs \xi, \v\right> + \E \max_{\w \in T^{p-1}} \left<\bs \zeta, \w\right>\nn\\
  & = & \E \|\bs \xi\|_2 + \E \|\bs \zeta\|_\infty\nn\\
  &\leq & \sqrt{n} + \sqrt{\log p}.
\end{eqnarray}
Here we have used
\begin{eqnarray}
  \E \|\bs \xi\|_2 \leq \sqrt{\E \|\bs \xi\|_2^2} = \sqrt{n}
\end{eqnarray}
due to Jensen's inequality, and
\begin{eqnarray}
  \E \|\bs \zeta\|_\infty = \E \max_i {\bs \zeta}_i \leq \sqrt{\log p},
\end{eqnarray}
a fact given by \cite[Equation 3.13, page 79]{ledoux1991probability}.

The conclusion of Theorem \ref{thm:randomcmsv} then follows from \eqref{eqn:estimate_gwidth}, \eqref{eqn:estimate_ginf}, and suitable choice of $c_1$. \qed
\end{proof}

\subsection{Proof of Proposition \ref{pro:max_sparse}}\label{app:pf:max_sparse}
\begin{proof}
We rewrite the optimization \eqref{eqn:s_star_larger} as
\begin{eqnarray}\label{eqn:max_inf_Q_1}
  \frac{1}{s^*} = \max_{\z}\|\z\|_{\b \infty}\text{\ s.t. \ } Q\z = 0, \|\z\|_{\b 1} \leq 1.
\end{eqnarray}
Note that in \eqref{eqn:max_inf_Q_1}, we are maximizing a convex function over a convex set, which is in general very difficult. We will use a relaxation technique to compute an upper bound on the optimal value of \eqref{eqn:max_inf_Q_1}. Define a matrix variable $P$ of the same size as $Q$. Since the dual norm of $\|\cdot\|_{\b \infty}$ is $\|\cdot\|_{\b 1}$, we have
\begin{eqnarray}\label{eqn:relax}
  &&\max_{\z}\left\{\|\z\|_{\b \infty}: \|\z\|_{\b 1} \leq 1, Q\z = 0\right\}\nonumber\\
&=&\max_{\u,\z}\left\{\u^T\z: \|\z\|_{\b 1} \leq 1, \|\u\|_{\b 1} \leq 1, Q\z = 0\right\}\nonumber\\
&=& \max_{\u,\z}\left\{\u^T(\z-P^TQ\z): \|\z\|_{\b 1} \leq 1, \|\u\|_{\b 1} \leq 1, Q\z = 0\right\}\nonumber\\
&\leq & \max_{\u,\z}\left\{\u^T(\I_{np}-P^TQ)\z: \|\z\|_{\b 1} \leq 1, \|\u\|_{\b 1} \leq 1\right\}.
\end{eqnarray}
In the last expression, we have dropped the constraint $Q\z = 0$. Note that the unit ball $\{\z: \|\z\|_{\b 1} \leq 1\}\subset \R^{np}$ is the convex hull of $\{\e_i^p\otimes \v: 1\leq i \leq p, \v\in \R^n, \|\v\|_2 \leq 1\}$ and $\u^T(\I_{np}-P^TQ)\z$ is convex (actually, linear) in $\z$. As a consequence, we have
\begin{eqnarray}\label{eqn:upperbound_unitball}
&&\max_{\u,\z}\left\{\u^T(\I_{np}-P^TQ)\z: \|\z\|_{\b 1} \leq 1, \|\u\|_{\b 1} \leq 1\right\}\nonumber\\
&=&\max_{j, \v,\u}\left\{\u^T(\I_{np}-P^TQ)(\e_i^p\otimes \v): \|\u\|_{\b 1} \leq 1, \|\v\|_2 \leq 1\right\}\nonumber\\
&=& \max_{j}\max_{\v,\u}\left\{\u^T(\I_{np}-P^TQ)_j\v: \|\u\|_{\b 1} \leq 1, \|\v\|_{2} \leq 1\right\}\nonumber\\
&=& \max_{j}\max_{\u}\left\{\|(\I_{np}-P^TQ)_j^T\u\|_2: \|\u\|_{\b 1} \leq 1\right\},
\end{eqnarray}
where $(\I_{np}-P^TQ)_j$ denotes the $j$th column blocks of $\I_{np}-P^TQ$, namely, the submatrix of $\I_{np}-P^TQ$ formed by the $((j-1)n+1)$th to $jn$th columns.

Applying the same argument to the unit ball $\{\u: \|\u\|_{\b 1} \leq 1\}$ and the convex function $\|(\I_{np}-P^TQ)_j^T\u\|_2$, we obtain
\begin{eqnarray}
&&\max_{\u,\z}\left\{\u^T(\I_{np}-P^TQ)\z: \|\z\|_{\b 1} \leq 1, \|\u\|_{\b 1} \leq 1\right\}\nonumber\\
&=& \max_{i, j}\|(\I_{np}-P^TQ)_{i,j}\|_2.
\end{eqnarray}
Here $(\I_{np}-P^TQ)_{i,j}$ is the submatrix of $\I_{np}-P^TQ$ formed by the $((i-1)n+1)$th to $in$th rows and the $((j-1)n+1)$th to $jn$th columns, and $\|\cdot\|_2$ is the spectral norm (the largest singular value).

Since $P$ is arbitrary, the tightest upper bound is obtained by minimizing $\max_{i, j}\|(\I_{np}-P^TQ)_{i,j}\|_2$ with respect to $P$:
\begin{eqnarray}
  1/s^* &= & \max_{\z}\left\{\|\z\|_{\b \infty}: \|\z\|_{\b 1} \leq 1, Q\z = 0\right\}\nonumber\\
&\leq & \min_{P}\max_{i, j}\|(\I_{np}-P^TQ)_{i,j}\|_2\nonumber\\
& \df & 1/s_*.
\end{eqnarray}

Partition $P$ and $Q$ as $P = \left[P_1,\ldots,P_p\right]$ and $Q = \left[Q_1,\ldots,Q_p\right]$, with $P_i$ and $Q_j$ having $n$ columns each. We explicitly write
\begin{eqnarray}
(\I_{np}-P^TQ)_{i,j} = \delta_{ij}\I_n - P_i^TQ_j,
\end{eqnarray}
where $\delta_{ij} = 1$ for $i = j$ and $0$ otherwise. As a consequence, we obtain
\begin{eqnarray}
  &&\min_{P}\max_{i, j}\|(\I_{np}-P^TQ)_{i,j}\|_2\nonumber\\
   &=&\min_{P_1,\ldots,P_p}\max_{i}\max_j\|\delta_{ij}\I_n - P_i^TQ_j\|_2\nonumber\\
    &=&\max_i\min_{P_i}\max_{j}\|\delta_{ij}\I_n - P_i^TQ_j\|_2.
\end{eqnarray}
We have moved the $\max_i$ to the outmost because for each $i$, $\max_{j}\|\delta_{ij}\I_n - P_i^TQ_j\|_2$ is a function of only $P_i$ and does not depends on other variables $P_l, l\neq i$. \qed
\end{proof}

\subsection{Proof of Proposition \ref{pro:fix_fs}}\label{app:pf:fix_fs}
\begin{proof}
\begin{enumerate}
\item Since in the optimization problem defining $f_s(\eta)$, the objective function $\|\z\|_{\b \infty}$ is continuous, and the constraint correspondence
\begin{eqnarray}
  C(\eta): [0, \infty) &\twoheadrightarrow& \R^{np}\nonumber\\
  \eta &\mapsto& \left\{\z: \|Q\z\|_\diamond \leq 1, {\|\z\|_{\b 1}} \leq s \eta\right\}
\end{eqnarray}
is compact-valued and continuous (both upper and lower hemicontinuous), according to Berge's Maximum Theorem \cite{Berge1997maximum} the optimal value function $f_s(\eta)$ is continuous.

\item The monotone (non-strict) increasing property is obvious as increasing $\eta$ enlarges the region over which the maximization is taken. We now show the strict increasing property. Suppose $0< \eta_1 < \eta_2$, and $f_s(\eta_1)$ is achieved by $\z_1^* \neq 0$, namely, $f_s(\eta_1) = \|\z_1^*\|_{\b \infty}$, $\|Q\z_1^*\|_\diamond \leq 1$, and $\|\z_1^*\|_{\b 1} \leq s \eta_1$. If $\|Q\z_1^*\|_\diamond < 1$ (this implies $\|\z_1^*\|_{\b 1} = s\eta_1$), we define $\z_2 = c \z_1^*$ with $c = \min\left(1/\|Q\z_1^*\|_\diamond, s\eta_2/\|\z_1^*\|_{\b 1}\right) > 1$. We then have $\|Q\z_2\|_\diamond \leq 1$, $\|\z_2\|_{\b 1} \leq s\eta_2$, and $f_s(\eta_2) \geq \|\z_2\|_{\b\infty} = c\|\z_1^*\|_{\b\infty} > f_s(\eta_1)$.

Consider the remaining case that $\|Q\z_1^*\|_\diamond = 1$. Without loss of generality, suppose $\|\z_1^*\|_{\b \infty} = \max_{1\leq j \leq p} \|\z_{1j}^*\|_2$ is achieved by the block $\z_{11}^* \neq 0$. Since $Q_1\z_{11}^*$ is linearly dependent with the columns of $\{Q_j, j =2, \ldots, p\}$ ($m$ is much less than $(p-1)n$), there exist $\{\bs \alpha_j \in \R^n\}_{j=2}^p$ such that $Q_1\z_{11}^* + \sum_{j=2}^p Q_j\bs \alpha_j = 0$. Define $\bs \alpha = \left[\z_{11}^{*T}, \bs \alpha_2^T, \cdots, \bs \alpha_p^T\right] \in \R^{np}$ satisfying $Q\bs \alpha = 0$ and $\z_2 = \z_1^* + c \bs \alpha$ for $c > 0$ sufficiently small such that $\|\z_2\|_{\b 1} \leq \|\z_1^*\|_{\b 1} + c \|\bs \alpha\|_{\b 1} \leq s\eta_1 + c\|\bs \alpha\|_{\b 1} \leq s\eta_2$. Clearly, $\|Q\z_2\|_\diamond = \|Q\z_1^* + cQ\bs \alpha\|_\diamond = \|Q\z_1^*\|_\diamond = 1$. As a consequence, we have
$f_s(\eta_2) \geq \|\z_2\|_{\b \infty} \geq \|\z_{21}\|_2 = (1+c)\|\z_{11}^*\|_2 > \|\z_1^*\|_{\b \infty} = f_s(\eta_1)$.

The case for $\eta_1 = 0$ is proved by continuity.

\item Next we show $f_s(\eta) > s\eta$ for sufficiently small $\eta > 0$. Take $\z$ as the vector whose first element is $s\eta$ and zero otherwise. We have $\|\z\|_{\b 1} = s\eta$ and $\|\z\|_{\b \infty} = s\eta > \eta$ (recall $s \in (1, s^*)$). In addition, when $\eta > 0$ is sufficiently small, we also have $\|Q\z\|_\diamond \leq 1$. Therefore, for sufficiently small $\eta$, we have $f_s(\eta) \geq s\eta > \eta$.

We next prove the existence of $\eta_B > 0$ and $\rho_B \in (0, 1)$ such that
\begin{eqnarray}\label{eqn:largeeta}
  f_s(\eta) < \rho_B \eta,\ \forall\  \eta > \eta_B.
\end{eqnarray}
We use contradiction to prove this statement. Suppose for all $\eta_B > 0$ and $\rho_B \in (0, 1)$, there exists $\eta > \eta_B$ such that $f_s(\eta) \geq \rho_B\eta$. Construct sequences $\{\eta^{(k)}\}_{k=1}^\infty \subset (0, \infty)$, $\{\rho^{(k)}\}_{k=1}^\infty \subset (0, 1)$, and $\{\z^{(k)}\}_{k=1}^\infty \subset \R^{np}$ such that
\begin{eqnarray}
&&\lim_{k\rightarrow \infty}\eta^{(k)} = \infty,\nonumber\\
&&\lim_{k\rightarrow \infty}\rho^{(k)} = 1, \nonumber\\
&&\rho^{(k)}\eta^{(k)} \leq f_s(\eta^{(k)}) = \|\z^{(k)}\|_{\b \infty},\nonumber\\
&&\|Q\z^{(k)}\|_\diamond \leq 1,\nonumber\\
&&\|\z^{(k)}\|_{\b 1} \leq s \eta^{(k)}.
\end{eqnarray}
Decompose $\z^{(k)} = \z^{(k)}_1 + \z^{(k)}_2$, where $\z^{(k)}_1$ is in the null space of $Q$ and $\z^{(k)}_2$ in the orthogonal complement of the null space of $Q$. The sequence $\{\z_2^{(k)}\}_{k=1}^\infty $ is bounded since $c \|\z_2^{(k)}\| \leq \|Q\z_2^{(k)}\|_\diamond \leq 1$ where $c = \inf_{\z: \z \perp \mathrm{null}(Q)} \|Q\z\|_\diamond/\|\z\| > 0$ and $\|\cdot\|$ is any norm. Then $\infty = \lim_{k\rightarrow \infty} \|\z^{(k)}\|_{\b \infty} \leq  \lim_{k\rightarrow \infty} (\|\z_1^{(k)}\|_{\b \infty}+ \|\z_2^{(k)}\|_{\b \infty})$ implies $\{\z_1^{(k)}\}_{k=1}^\infty$ is unbounded. For sufficiently large $k$, we proceed as follows:
\begin{eqnarray}
 && s^* > \frac{s}{\left(\frac{s}{s^*}\right)^{1/4}} \geq \frac{s\eta^{(k)}}{\rho^{(k)}\eta^{(k)}} \geq \frac{\|\z^{(k)}\|_{\b 1}}{\|\z^{(k)}\|_{\b \infty}}\nonumber\\
 &&\ \ \ \ \ \ \ \ \ \ \ \ \ \ \geq  \frac{\|\z^{(k)}_1\|_{\b 1} - \|\z_2^{(k)}\|_{\b 1}}{\|\z_1^{(k)}\|_{\b \infty}+\|\z_2^{(k)}\|_{\b \infty}} \geq \left(\frac{s}{s^*}\right)^{1/4} \frac{\|\z_1^{(k)}\|_{\b 1}}{\|\z_1^{(k)}\|_{\b \infty}},
\end{eqnarray}
where the second and last inequalities hold only for sufficiently large $k$ and the last inequality is due to the unboundedness of $\{\z_1^{(k)}\}$ and boundedness of $\{\z_2^{(k)}\}$. As a consequence, we have
\begin{eqnarray}
 \frac{\|\z_1^{(k)}\|_{\b 1}}{\|\z_1^{(k)}\|_{\b \infty}} \leq \frac{s}{\sqrt{\frac{s}{s^*}}} = \sqrt{ss^*} < s^* \text{\ with \ } Q\z_1^{(k)} = 0,
\end{eqnarray}
which contradicts the definition of $s^*$.

\item Next we show $f_s(\eta)$ has a unique positive fixed point $\eta^*$, which is equal to $\gamma^* \df 1/\omega_\diamond(Q,s)$. Properties 1) and 3) imply that there must be at least one fixed point.

    To show the uniqueness, we first prove $\gamma^* \geq \eta^*$ for any fixed point $\eta^* = f_s(\eta^*)$. Suppose $\z^*$ achieves the optimal value of the optimization defining $f_s(\eta^*)$, \emph{i.e.},
\begin{eqnarray}
\eta^* = f_s(\eta^*) = \|\z^*\|_{\b \infty}, \|Q\z^*\|_\diamond \leq 1, \|\z^*\|_{\b 1} \leq s\eta^*.
\end{eqnarray}
Since $\|\z^*\|_{\b 1}/\|\z^*\|_{\b \infty} \leq s\eta^*/\eta^* \leq s$, we have
\begin{eqnarray}
  \gamma^* &\geq& \frac{\|\z^*\|_{\b \infty}}{\|Q\z^*\|_\diamond} \geq \eta^*.
\end{eqnarray}

If $\eta^* < \gamma^*$, we define $\eta_0 = (\eta^* + \gamma^*)/2$ and
\begin{eqnarray}\label{eqn:defrho}
\hskip -1cm &&\z^{\mathrm{c}} = \mathrm{argmax}_{\z}{\frac{s\|\z\|_{\b \infty}}{\|\z\|_{\b 1}}} \text{\ s.t. \ } \|Q\z\|_\diamond \leq 1, \|\z\|_{\b \infty} \geq \eta_0, \text{\ and}\nn\\
\hskip -1cm  &&\rho = {\frac{s\|\z^{\mathrm{c}}\|_{\b \infty}}{\|\z^{\mathrm{c}}\|_{\b 1}}}.
\end{eqnarray}
Suppose $\z^{**}$ with $\|Q\z^{**}\|_\diamond = 1$ achieves the optimal value of the optimization defining $\gamma^* = 1/\omega_\diamond(Q,s)$. Clearly, $\|\z^{**}\|_{\b \infty} = \gamma^* > \eta_0$, which implies $\z^{**}$ is a feasible point of the optimization defining $\z^{\mathrm{c}}$ and $\rho$. As a consequence, we have
\begin{eqnarray}
  \rho \geq {\frac{s\|\z^{**}\|_{\b \infty}}{\|\z^{**}\|_{\b 1}}} \geq 1.
\end{eqnarray}

\begin{figure*}[h!t]
\hskip -0cm
\centering
\includegraphics[width = 0.5\textwidth, trim = 0mm 0mm 0mm 0mm, clip]{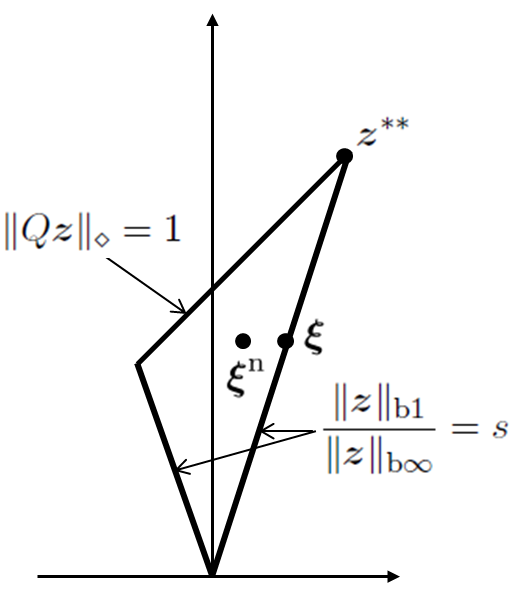}
\caption{Illustration of the proof for $\rho > 1$.}
\label{fig:proof}
\end{figure*}%

Actually we will show that $\rho > 1$. If $\|\z^{**}\|_{\b 1} < s\|\z^{**}\|_{\b \infty}$, we are done. If not (\emph{i.e.}, $\|\z^{**}\|_{\b 1} = s\|\z^{**}\|_{\b \infty}$), as illustrated in Figure \ref{fig:proof}, we consider $\bs \xi = \frac{\eta_0}{\gamma^*}\z^{**}$, which satisfies
\begin{eqnarray}
  &&\|Q\bs \xi\|_\diamond \leq \frac{\eta_0}{\gamma^*} < 1,\\
  &&\|\bs \xi\|_{\b \infty}= \eta_0, \text{\  and \ }\\
  &&\|\bs \xi\|_{\b 1} = s\eta_0.
\end{eqnarray}

To get ${\bs \xi}^{\mathrm{n}}$ as shown in Figure \ref{fig:proof}, pick the block of $\bs \xi$ with smallest non-zero $\ell_2$ norm, and scale that block by a small positive constant less than $1$. Because $s > 1$, $\bs \xi$ has more than one non-zero blocks, implying $\|{\bs \xi}^{\mathrm{n}}\|_{\b \infty}$ will remain the same. If the scaling constant is close enough to $1$, $\|Q{\bs \xi}^{\mathrm{n}}\|_\diamond$ will remain less than 1. But the good news is that $\|{\bs \xi}^{\mathrm{n}}\|_{\b 1}$ decreases, and hence $\rho \geq \frac{s\|{\bs \xi}^{\mathrm{n}}\|_{\b \infty}}{\|{\bs \xi}^{\mathrm{n}}\|_{\b 1}}$ becomes greater than 1.

Now we proceed to obtain the contradiction that $f_s(\eta^*) > \eta^*$. If $\|\z^{\mathrm{c}}\|_{\b 1} \leq s\cdot \eta^*$, then it is a feasible point of
\begin{eqnarray}\label{eqn:subt0}
  \max_{\z} \|\z\|_{\b \infty}\text{\ s.t. \ } \|Q\z\|_\diamond \leq 1, \|\z\|_{\b 1} \leq s\cdot \eta^*.
\end{eqnarray}
As a consequence, $f_s(\eta^*) \geq \|\z^{\mathrm{c}}\|_{\b \infty}\geq \eta_0 > \eta^*$, contradicting that $\eta^*$ is a fixed point, and we are done. If this is not the case, \emph{i.e.}, $\|\z^{\mathrm{c}}\|_{\b 1} > s\cdot \eta^*$, we define a new point
\begin{eqnarray}
  \z^{\mathrm{n}} = \tau \z^{\mathrm{c}}
\end{eqnarray}
with
\begin{eqnarray}
  \tau = \frac{s\cdot \eta^*}{\|\z^\mathrm{c}\|_{\b 1}} < 1.
\end{eqnarray}
Note that $\z^{\mathrm{n}}$ is a feasible point of the optimization defining $f_s(\eta^*)$ since
\begin{eqnarray}
&&\|Q\z^{\mathrm{n}}\|_\diamond = \tau \|Q\z^{\mathrm{c}}\|_\diamond < 1, \text{\ and \ }\\
&&\|\z^{\mathrm{n}}\|_{\b 1} = \tau \|\z^{\mathrm{c}}\|_{\b 1} = s\cdot \eta^*.
\end{eqnarray}
Furthermore, we have
\begin{eqnarray}
  \|\z^{\mathrm{n}}\|_{\b \infty}= \tau \|\z^{\mathrm{c}}\|_{\b \infty}= \rho \eta^*.
\end{eqnarray}
As a consequence, we obtain
\begin{eqnarray}
  f_s(\eta^*) &\geq& \rho \eta^* > \eta^*.
\end{eqnarray}
\begin{figure}[h!t]
\hskip -0cm
\centering
\includegraphics[width = 0.6\textwidth, trim = 0mm 0mm 0mm 0mm, clip]{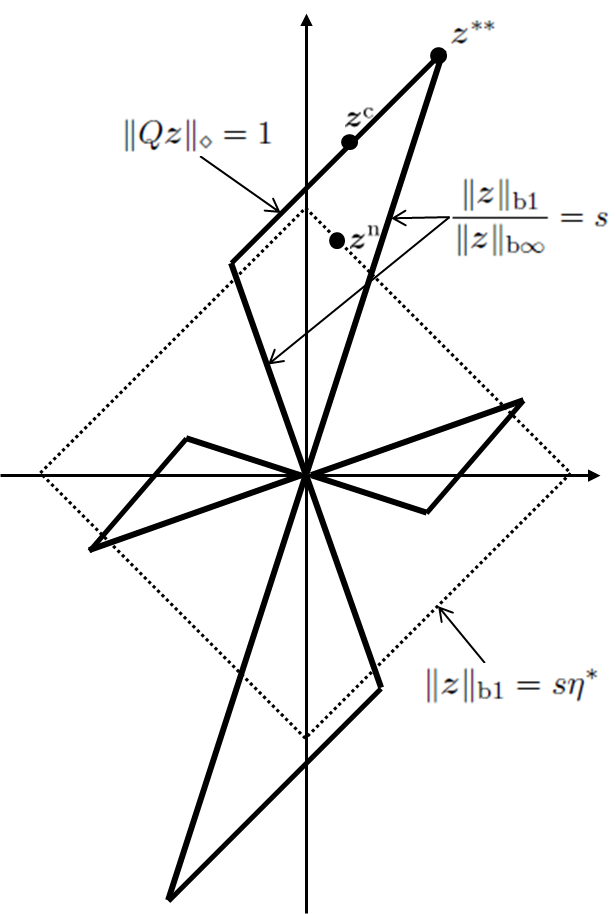}
\caption{Illustration of the proof for $f_s(\eta^*) \geq \rho\eta^*$.}
\label{fig:proof1}
\end{figure}%

Therefore, for any positive fixed point $\eta^*$, we have $\eta^* = \gamma^*$, \emph{i.e.}, the positive fixed point is unique.

\item Property 5) is a consequence of 1), 3), and 4).

\item We demonstrate the existence of $\rho_2(\epsilon)$ only. The existence of $\rho_1(\epsilon)$ can be proved in a similar manner, and hence is omitted.

    We need to show that for fixed $\epsilon > 0$, there exists $\rho(\epsilon) < 1$ such that for any $\eta \geq (1+\epsilon) \eta^*$ we have
\begin{eqnarray}
  f_s(\eta) \leq \rho(\epsilon) \eta.
\end{eqnarray}

In view of \eqref{eqn:largeeta}, we need to prove the above statement only for $\eta \in [(1+\epsilon)\eta^*, \eta_B]$. We use contradiction. Suppose for any $\rho \in (0, 1)$ there exists $\eta \in [(1+\epsilon)\eta^*, \eta_B]$ such that $f_s(\eta) > \rho \eta$. Construct sequences $\{\eta^{(k)}\}_{k=1}^\infty \subset [(1+\epsilon)\eta^*, \eta_B]$ and $\{\rho^{(k)}\}_{k=1}^\infty \subset (0, 1)$ with
\begin{eqnarray}
&& \lim_{k\rightarrow \infty} \rho^{(k)} = 1, \text{\ and}\nonumber\\
&& f_s(\eta^{(k)}) > \rho^{(k)} \eta^{(k)}.
\end{eqnarray}
Due to the compactness of $[(1+\epsilon)\eta^*, \eta_B]$, there must exist a subsequence $\{\eta^{(k_l)}\}_{l=1}^\infty$ of $\{\eta^{(k)}\}_{k=1}^\infty$ such that $\lim_{l\rightarrow \infty}\eta^{(k_l)} = \eta_{\mathrm{lim}}$ for some $\eta_{\mathrm{lim}} \in [(1+\epsilon)\eta^*, \eta_B]$. As a consequence of the continuity of $f_s(\eta)$, we have

\begin{eqnarray}
f_s(\eta_{\mathrm{lim}}) =  \lim_{l\rightarrow \infty} f_s(\eta^{(k_l)}) \geq \lim_{l\rightarrow \infty} \rho^{(k_l)} \eta^{(k_l)} = \eta_{\mathrm{lim}}.
\end{eqnarray}
Again due to the continuity of $f_s(\eta)$ and the fact that $f_s(\eta) < \eta$ for $\eta > \eta_B$, there exists $\eta_c \in [\eta_{\mathrm{lim}}, \eta_B]$ such that
\begin{eqnarray}
  f_s(\eta_c) &=& \eta_c,
\end{eqnarray}
contradicting the uniqueness of the fixed point for $f_s(\eta)$.

The result implies that starting from any initial point below the fixed point, through the iteration $\eta_{t+1} = f_s(\eta_t)$ we can approach an arbitrarily small neighborhood of the fixed point exponentially fast.
\end{enumerate}
\qed
\end{proof}

\subsection{Proof of Proposition \ref{pro:relax_subproblem}}\label{app:pf:relax_subproblem}
\begin{proof}
Introducing an additional variable $\v$ such that $Q\z = \v$, we have
\begin{eqnarray}\label{eqn:medium_obj}
  &&\max_{\z} \left\{\|\z\|_{\b \infty}: \|\z\|_{\b 1} \leq s\eta, \|Q\z\|_\diamond \leq 1\right\}\nonumber\\
  &= & \max_{\z, \v} \left\{\|\z\|_{\b \infty}: \|\z\|_{\b 1} \leq s\eta, Q\z = \v, \|\v\|_\diamond \leq 1\right\}\nonumber\\
    &= & \max_{\z, \v, \u} \left\{\u^T\z: \|\z\|_{\b 1} \leq s\eta, Q\z = \v, \|\v\|_\diamond \leq 1, \|\u\|_{\b 1} \leq 1\right\}\nonumber\\
&= & \max_{\z, \v, \u} \left\{\u^T(\z-P^T(Q\z - \v)): \|\z\|_{\b 1} \leq s\eta, Q\z = \v, \|\v\|_\diamond \leq 1, \|\u\|_{\b 1} \leq 1\right\}\nonumber\\
&\leq & \max_{\z, \v, \u} \left\{\u^T(\I_{np}-P^TQ)\z + \u^T P^T\v: \|\z\|_{\b 1} \leq s\eta, \|\v\|_\diamond \leq 1, \|\u\|_{\b 1} \leq 1\right\},
\end{eqnarray}
where for the last inequality we have dropped the constraint $Q\z = \v$. Similar to \eqref{eqn:upperbound_unitball}, we bound from above the first term in the objective function of \eqref{eqn:medium_obj}
\begin{eqnarray}
&&\max_{\z, \u}\left\{\u^T(\I_{np}-P^TQ)\z: \|\z\|_{\b 1} \leq s\eta, \|\u\|_{\b 1} \leq 1\right\}\nonumber\\
&\leq & \max_{i,j}s\eta \|(\I_{np}-P^TQ)_{i,j}\|_2\nonumber\\
& = & \max_{i,j}s\eta \|\delta_{ij}\I_{n}-P_i^TQ_j\|_2.
\end{eqnarray}
Here $P_i$ (resp. $Q_j$) is the submatrix of $P$ (resp. $Q$) formed by the $((i-1)n+1)$th to $in$th columns (resp. $((j-1)n+1)$th to $jn$th columns). For the second term $\u^TP^T\v$ in the objective function of \eqref{eqn:medium_obj}, the definition of $\|\cdot\|_\diamond^*$, the dual norm of $\|\cdot\|_\diamond$, leads to
\begin{eqnarray}
&& \max_{\u, \v}\left\{\u^TP^T\v: \|\u\|_{\b 1} \leq 1, \|\v\|_\diamond \leq 1\right\}\nonumber\\
&\leq & \max_{\u} \left\{\|P\u\|_\diamond^*: \|\u\|_{\b 1} \leq 1\right\}.
\end{eqnarray}
Since $\|P\u\|_\diamond^*$ is a convex function of $\u$ and the extremal points of the unit ball $\{\u: \|\u\|_{\b 1} \leq 1\}$ are $\{\e_i^p\otimes \v: \|\v\|_2 \leq 1, \v \in \R^n\}$, the above inequality is further bounded from above as
\begin{eqnarray}
&&\max_{\u} \left\{\|P\u\|_\diamond^*: \|\u\|_{\b 1} \leq 1\right\}\nonumber\\
&\leq & \max_i \max_{\v}\{\|P_i\v\|_\diamond^*: \|\v\|_2 \leq 1\} = \max_{i} \|P_i\|_{2,*}.
\end{eqnarray}
Here $\|P_i\|_{2,*}$ is the operator norm of $P_i$ from the space with norm $\|\cdot\|_2$ to the space with norm $\|\cdot\|_\diamond^*$. A tight upper bound on $\max_{\z} \left\{\|\z\|_{\b \infty}: \|\z\|_{\b 1} \leq s\eta, \|Q\z\|_\diamond \leq 1\right\}$ is obtained by minimizing with respect to $P$:
\begin{eqnarray}\label{eqn:generaldualbd}
  &&\max_{\z} \left\{\|\z\|_{\b \infty}: \|\z\|_{\b 1} \leq s\eta, \|Q\z\|_\diamond \leq 1\right\}\nonumber\\
  &\leq & \min_{P}\max_{i, j} s\eta\|\delta_{ij}\I_{n}-P_i^TQ_j\|_2 + \|P_i\|_{2,*}\nonumber\\
    &= & \max_i \min_{P_i}\max_{j} s\eta\|\delta_{ij}\I_{n}-P_i^TQ_j\|_2 + \|P_i\|_{2,*}.
\end{eqnarray}
When $Q = A$ and $\diamond = 2$, the operator norm
\begin{eqnarray}\label{eqn:dual2}
\|P_i\|_{2,*} = \|P_i\|_{2,2} = \|P_i\|_2.
\end{eqnarray}
When $Q = A^TA$ and $\diamond = \b \infty$, we have
\begin{eqnarray}\label{eqn:dualbinf}
&& \|P_i\|_{2,*} = \|P_i\|_{2,\b 1} = \max_{\v: \|\v\|_2 \leq 1} \|P_i\v\|_{\b 1}\nonumber\\
&=& \max_{\v:\|\v\|_2 \leq 1} \sum_{l=1}^p \|P_i^l\v\|_2 \leq \sum_{l = 1}^p\|P_i^l\|_2.
\end{eqnarray}
Here $P_i^l$ is the submatrix of $P$ formed by the $((i-1)n+1)$th to $in$th columns and the $((l-1)n+1)$th to $ln$th rows. Proposition \ref{pro:relax_subproblem} then follows from \eqref{eqn:generaldualbd}, \eqref{eqn:dual2}, and \eqref{eqn:dualbinf}. \qed
\end{proof}

\subsection{Proof of Proposition \ref{pro:propertygs}}\label{app:pf:propertygs}
\begin{proof}
\begin{enumerate}
\item First note that adding an additional constraint $\|P_i\|_2 \leq s\eta$ in the definition of $g_{s, i}$ does not change the definition, because $g_{s,i}(\eta) \leq s\eta$ as easily seen by setting $P_i = 0$:
\begin{eqnarray}
  g_{s, i} &=& \min_{P_i} \{s\eta  (\max_j \|\delta_{ij}\I_n - P_i^T Q_j\|_2) + \|P_i\|_2: \|P_i\|_2 \leq s\eta\}\nonumber\\
  & = & -\max_{P_i} \{-s\eta  (\max_j \|\delta_{ij}\I_n - P_i^T Q_j\|_2) - \|P_i\|_2: \|P_i\|_2 \leq s\eta\}.
\end{eqnarray}
Since the objective function to be maximized is continuous, and the constraint correspondence
\begin{eqnarray}
  C(\eta) = \left\{P_i: {\|P_i\|_2} \leq s \eta\right\}
\end{eqnarray}
is compact-valued and continuous (both upper and lower hemicontinuous), according to Berge's Maximum Theorem \cite{Berge1997maximum} the optimal value function $g_{s, i}(\eta)$ is continuous. The continuity of $g_s(\eta)$ follows from the fact that finite maximization preserves the continuity.

\item To show the strict increasing property, suppose $\eta_1 < \eta_2$ and $P_{i,2}^*$ achieves $g_{s,i}(\eta_2)$. Then we have
\begin{eqnarray}
  g_{s,i}(\eta_1) &\leq&  s\eta _1 \left(\max_j \|\delta_{ij}\I_n - P_{i2}^{*T} Q_j\|_2\right) + \|P_{i2}^*\|_2\nonumber\\
  &<& s\eta _2 \left(\max_j \|\delta_{ij}\I_n - P_{i2}^{*T} Q_j\|_2\right) + \|P_{i2}^*\|_2\nonumber\\
  &=& g_{s,i}(\eta_2).
\end{eqnarray}
The strict increasing of $g_s(\eta)$ then follows immediately.

\item The concavity of $g_{s,i}(\eta)$ follows from the fact that $g_{s,i}(\eta)$ is the minimization of a function of variables $\eta$ and $P_i$, and when $P_i$, the variable to be minimized, is fixed, the function is linear in $\eta$.

\item Next we show that when $\eta > 0$ is sufficiently small $g_s(\eta) \geq s \eta$. For any $i$, we have the following,
    \begin{eqnarray}
g_{s,i}(\eta) &=& \min_{P_i} s\eta  \left(\max_j \|\delta_{ij}\I_n - P_i^T Q_j\|_2\right) + \|P_i\|_2 \nonumber\\
&\geq & \min_{P_i} s\eta  \left(1 -  \|P_i^T Q_i\|_2\right) + \|P_i\|_2\nonumber\\
&\geq & \min_{P_i} s\eta  \left(1 -  \|P_i\|_2\|Q_i\|_2\right) + \|P_i\|_2\nonumber\\
& = & s\eta + \min_{P_i} \|P_i\|_2\left(1 -  s\eta\|Q_i\|_2\right)\nonumber\\
& \geq & s\eta > \eta,
    \end{eqnarray}
where the minimum of the last optimization is achieved at $P_i = 0$ when $\eta < 1/(s\|Q_i\|_2)$. Clearly, $g_s(\eta) = \max_i g_{s,i}(\eta) \geq s\eta > \eta$ for such $\eta$.

Recall that
\begin{eqnarray}
  \frac{1}{s_*} &=&  \max_i\min_{P_i}\max_{j}\|\delta_{ij}\I_n - P_i^TQ_j\|_2.
\end{eqnarray}
Suppose $P_i^*$ is the optimal solution for each $\min_{P_i}\max_{j}\|\delta_{ij}\I_n - P_i^TQ_j\|_2$. For each $i$, we then have
\begin{eqnarray}
  \frac{1}{s_*} &\geq & \max_j \|\delta_{ij}\I_n - P_i^{*T}Q_j\|_2,
\end{eqnarray}
which implies
\begin{eqnarray}
  g_{s,i}(\eta) &=& \min_{P_i} s\eta  \left(\max_j \|\delta_{ij}\I_n - P_i^T Q_j\|_2\right) + \|P_i\|_2 \nonumber\\
  &\leq & s\eta  \left(\max_j \|\delta_{ij}\I_n - P_i^{*T} Q_j\|_2\right) + \|P_i^*\|_2 \nonumber\\
  &\leq & \frac{s}{s_*} \eta + \|P_i^*\|_2.
\end{eqnarray}

As a consequence, we obtain
\begin{eqnarray}
g_s(\eta) = \max_i g_{s,i}(\eta) \leq \frac{s}{s_*} \eta + \max_i \|P_i^*\|_2.
\end{eqnarray}
Pick $\rho \in (s/s_*, 1)$. Then, we have the following when $\eta > \max_i \|P_i^*\|_2/(\rho - s/s_*)$:
\begin{eqnarray}
  g_s(\eta) \leq \rho \eta.
\end{eqnarray}

\item We first show the existence and uniqueness of the positive fixed points for $g_{s,i}(\eta)$. Properties 1) and 4) imply that $g_{s,i}(\eta)$ has at least one positive fixed point. (Interestingly, 2) and 4) also imply the existence of a positive fixed point, see \cite{tarski1955fixedpoint}.) To prove uniqueness, suppose there are two fixed points $0 < \eta_1^* < \eta_2^*$. Pick $\eta_0$ small enough such that $g_{s,i}(\eta_0) > \eta_0 > 0$ and $\eta_0 < \eta_1^*$. Then $\eta_1^* = \lambda \eta_0 + (1-\lambda)\eta_2^*$ for some $\lambda \in (0, 1)$, which implies that $g_{s,i}(\eta_1^*) \geq \lambda g_{s,i}(\eta_0) + (1-\lambda) g_{s,i}(\eta_2^*) > \lambda \eta_0 + (1-\lambda) \eta_2^* = \eta_1^*$ due to the concavity, contradicting $\eta_1^* = g_{s,i}(\eta_1^*)$.

    The set of positive fixed points for $g_s(\eta)$, $\{\eta \in (0, \infty): \eta = g_s(\eta) = \max_i g_{s,i}(\eta)\}$, is a subset of $\bigcup_{i=1}^p \{\eta \in (0, \infty): \eta = g_{s,i}(\eta) \} = \{\eta_i^*\}_{i=1}^p$. We argue that
    \begin{eqnarray}
     \eta^* = \max_i \eta_i^*
    \end{eqnarray}
is the unique positive fixed point for $g_s(\eta)$.

We proceed to show that $\eta^*$ is a fixed point of $g_s(\eta)$. Suppose $\eta^*$ is a fixed point of $g_{s,i_0}(\eta)$. Then it suffices to show that $g_s(\eta^*) = \max_i g_{s,i}(\eta^*) = g_{s,i_0}(\eta^*)$. If this is not the case, there exists $i_1 \neq i_0$ such that $g_{s,i_1}(\eta^*) > g_{s,i_0}(\eta^*) = \eta^*$. The continuity of $g_{s,i_1}(\eta)$ and the property 4) imply that there exists $\eta > \eta^*$ with $g_{s,i_1}(\eta) = \eta$, contradicting the definition of $\eta^*$.

To show the uniqueness, suppose $\eta_1^*$ is fixed point of $g_{s,i_1}(\eta)$ satisfying $\eta_1^* < \eta^*$. Then, we must have $g_{s,i_0}(\eta_1^*) > g_{s,i_1}(\eta_1^*)$, because otherwise the continuity implies the existence of another fixed point of $g_{s,i_0}(\eta)$. As a consequence, $g_s(\eta_1^*) > g_{s,i_1}(\eta_1^*) = \eta_1^*$ and $\eta_1^*$ is not a fixed point of $g_s(\eta)$.
\item This property simply follows from the continuity, the uniqueness, and property 4).
\item We use contradiction to show the existence of $\rho_1(\epsilon)$ in 7). In view of 4), we need only to show the existence of such a $\rho_1(\epsilon)$ that works for $\eta_L \leq \eta \leq (1-\epsilon)\eta^*$ where $\eta_L = \sup\{\eta: g_s(\xi) > s\xi, \forall 0 < \xi \leq \eta\}$. Supposing otherwise, we then construct sequences $\{\eta^{(k)}\}_{k=1}^\infty \subset [\eta_L, (1-\epsilon)\eta^*]$ and $\{\rho_1^{(k)}\}_{k=1}^\infty \subset (1, \infty)$ with
\begin{eqnarray}
&& \lim_{k\rightarrow \infty} \rho_1^{(k)} = 1,\text{\ and}\nonumber\\
&& g_s(\eta^{(k)}) \leq \rho^{(k)} \eta^{(k)}.
\end{eqnarray}
Due to the compactness of $[\eta_L, (1-\epsilon)\eta^*]$, there must exist a subsequence $\{\eta^{(k_l)}\}_{l=1}^\infty$ of $\{\eta^{(k)}\}$ such that $\lim_{l\rightarrow \infty}\eta^{(k_l)} = \eta_{\mathrm{lim}}$ for some $\eta_{\mathrm{lim}} \in [\eta_L, (1-\epsilon)\eta^*]$. As a consequence of the continuity of $g_s(\eta)$, we have

\begin{eqnarray}
g_s(\eta_{\mathrm{lim}}) =  \lim_{l\rightarrow \infty} g_s(\eta^{(k_l)}) \leq \lim_{l\rightarrow \infty} \rho_1^{(k_l)} \eta^{(k_l)} = \eta_{\mathrm{lim}}.
\end{eqnarray}
Again due to the continuity of $g_s(\eta)$ and the fact that $g_s(\eta) < \eta$ for $\eta < \eta_L$, there exists $\eta_c \in [\eta_L, \eta_{\mathrm{lim}}]$ such that
\begin{eqnarray}
  g_s(\eta_c) &=& \eta_c,
\end{eqnarray}
contradicting the uniqueness of the fixed point for $g_s(\eta)$. The existence of $\rho_2(\epsilon)$ can be proved in a similar manner.
\end{enumerate}
\qed
\end{proof}
\bibliographystyle{spmpsci}
\bibliography{/Users/gongguotang/SugarSync/Papers/Material/BibTex/IEEEabrv,/Users/gongguotang/SugarSync/Papers/Material/BibTex/Gongbib}
\end{document}